\documentclass[11pt]{article}

\usepackage[utf8]{inputenc}
\usepackage[OT1]{fontenc}

\usepackage{graphicx}

\usepackage{amsmath, amsthm, amssymb}
\theoremstyle{plain}
\newtheorem{lemma}{Lemma}[section]
\newtheorem{proposition}{Proposition}[section]
\newtheorem{theorem}{Theorem}[section]
\theoremstyle{definition}
\newtheorem{example}{Example}

\usepackage{bm}
\usepackage{multirow}
\usepackage{color}

\usepackage{fullpage}
\usepackage{setspace}
\usepackage{csquotes}
\MakeOuterQuote{"}

\usepackage[normalem]{ulem}

\usepackage{natbib}

\usepackage{authblk}

\title{Analyzing order flows in limit order books \\ with ratios of Cox-type intensities}

\author[1]{Ioane Muni Toke\thanks{Laboratoire MICS et Chaire de Finance Quantitative, CentraleSupélec, Bâtiment Bouygues, 3 rue Joliot Curie, 91190 Gif-sur-Yvette, France. ioane.muni-toke@centralesupelec.fr}}
\author[2]{Nakahiro Yoshida\thanks{Graduate School of Mathematical Sciences, University of Tokyo: 3-8-1 Komaba, Meguro-ku, Tokyo 153- 8914, Japan. nakahiro@ms.u-tokyo.ac.jp}}
\affil[1]{Mathématiques et Informatique pour la Complexité et les Systèmes, CentraleSupélec, Université Paris-Saclay, France}
\affil[2]{Graduate School of Mathematical Sciences, University of Tokyo, Japan.}
\affil[1,2]{CREST, Japan Science and Technology Agency, Japan.}

\begin{document}

\maketitle

\begin{abstract}
We introduce a Cox-type model for relative intensities of orders flows in a limit order book. The model assumes that all intensities share a common baseline intensity, which may for example represent the global market activity. Parameters can be estimated by quasi-likelihood maximization, without any interference from the baseline intensity. Consistency and asymptotic behavior of the estimators are given in several frameworks, and model selection is discussed with information criteria and penalization. The model is well-suited for high-frequency financial data: fitted models using easily interpretable covariates show an excellent agreement with empirical data. Extensive investigation on tick data consequently helps identifying trading signals and important factors determining the limit order book dynamics. We also illustrate the potential use of the framework for out-of-sample predictions.
\end{abstract}


\section{Introduction}

The limit order book is the central structure that aggregates all orders submitted by market participants to buy or sell a given asset on a financial market. Buy offers form the bid side, sell offers form the ask side. Simplified representations of possible interactions with a limit order book usually consider three archetypal sorts of orders : limit orders, market orders, and cancellations. Buy (resp. sell) limit orders are orders to buy (resp. sell) a given quantity of the asset with a given limit price strictly lower (resp. greater) than the current best ask (resp. bid) quote. Limit orders are stored in the book until matched or canceled. Buy (resp. sell) market orders are submitted without any limit price, and are thus matched against the current best ask (resp. bid) orders. Cancellations remove a non-executed limit order from the book.

The rise of electronic markets, the accelerating rate of trading on financial markets, the development of optimal trading strategies by brokers, etc., have pushed for better investigation and modeling of limit order books. \cite{Chakraborti2011}, \cite{Gould2013} and \cite{Abergel2016} provide some overview on recent modeling efforts. \cite{Cont2010} has been a seminal model using Poisson processes. \cite{Huang2015} investigate Markovian dynamics dependent on the size of the queues of the limit order book. \cite{MuniTokeYoshida2017} propose a state-dependent model that makes orders intensities depend on financial signals such as the bid-ask spread.
A common difficulty in estimating these models is the necessity to cope with the fact that market activity is highly fluctuating during the day, at all timescales. Daily market activity (for example measured in number or volume of trades, or in number or volume of all orders) is known to globally exhibit a U-shaped pattern, with a lower activity in the middle of the day, and a much higher activity in the morning after the market opening, and in the afternoon before market close. This seasonality effect is non-smooth: exogenous news (announcements of company results, of acquisitions, of macroeconomic indicators, etc.) occur all day long, at random or predetermined times, and may incur activity bursts. In Europe, openings of American markets are usually followed by some activity increase. 
It may be quite difficult to incorporate such variations in a model. In order to avoid such problems, one may try to remove the most hectic parts of the samples, and/or focus on a limited time interval, and/or split trading days in several parts. Examples can be found in, e.g., the growing literature on Hawkes processes in finance: \cite{Bacry2012} limits tests on their estimation procedure to two-hour periods to avoid seasonality effects; \cite{Lallouache2016} uses time-dependent piecewise-linear baseline intensity (with nodes spaced every few hours).
However, even at these scales, global market activity varies and may limit the reliability of estimation procedures.

In this work we continue previous investigations on the influences of financial variables (state of the order book or other trading signals) on the processes of order submissions. We assume that orders submissions are modeled by point processes with Cox-like intensities depending on given covariates. However, we do not directly try to model and fully estimate each and every intensities of the model, as in e.g. \cite{MuniTokeYoshida2017}, but we rather try to estimate the relative influences of given covariates on the intensities. To this end we assume that their exists a possibly random baseline intensity that is common to all processes under investigation, and that could for example incorporate the seasonality effects and changing activities that we have described. By dealing with ratios of intensities, we can then remove this baseline intensity from our estimation procedure, and thereby obtain a model that focuses only on the financial covariates under investigation.

In Section \ref{section:ModelDescription} we describe the general model of ratios of intensities. In Section \ref{section:LikelihoodAnalysis} we show that the quasi-likelihood estimators of the model are consistent and asymptotically normal. Section \ref{section:Penalization} discusses the method with respect to information criteria and penalization. Finally, Section \ref{section:EmpiricalResults} illustrates the benefits of the model with several examples of limit order book analysis. The ratio model is able to reproduce empirical observations on trading signals in orders flows such as imbalance, spread, quantities available, etc. Estimation results on more than 30 stocks in the Paris Stock Exchange for most of the year 2015 show a very good agreement between empirical observations and the proposed ratio model.

\section{Model description}
\label{section:ModelDescription}
Let $\mathbb I=\{0,\ldots,\bar i\}$ and $\mathbb J=\{1,\ldots,\bar j\}$ for some strictly positive integers $\bar i$ and $\bar j$. Let $(N^i_t)_{t\geq 0}, i\in\mathbb I,$  be some counting processes. We assume that the intensities $\lambda^i(t), i\in\mathbb I,$ of these counting processes share a common baseline intensity $\lambda_0(t)$. $\lambda_0(t)$ is neither observable nor specified as a function of observables and parameters. For any $i\in\mathbb I$, the intensity $\lambda^i(t)$ is written :
\begin{equation}
	\lambda^i(t,\vartheta) = \lambda_0(t) \exp\left( \sum_{j\in\mathbb J} \vartheta^i_j X_ j(t)\right),
	\label{eq:CoxModel}
\end{equation}
where $(X_j(t))_{t\geq 0}$ is the $j$-th observable covariate process and $\vartheta=(\vartheta^i_j)_{i\in\mathbb I, j\in\mathbb J}$ a parameter vector.

We are not interested in the value of the coefficient $\vartheta^i_j$, modeling the specific response of the counting process $N^i$ to the covariate $X^j$, but rather in the relative responses of the intensities compared to each other. 
Let
\begin{equation}
	\theta^i_j = \vartheta^i_j - \vartheta^0_j, \quad (i\in\mathbb I, j\in\mathbb J)
	\label{eq:DefinitionTheta}
\end{equation}
be the relative response of process $i$ compared to process $0$ with respect to covariate $j$, $j\in\mathbb J$.
Obviously, $\forall j\in\mathbb J, \theta^0_j=0$ since the process $0$ is taken as an arbitrary reference, and $\forall j\in\mathbb J, \forall (i,i')\in\mathbb I^2, \theta^{i'}_j-\theta^i_j = \vartheta^{i'}_j-\vartheta^i_j$, i.e. the differences between absolute and relative responses are equal.
Instead of the standard intensities defined at Equation \eqref{eq:CoxModel}, and since we are only interested in the relative responses, we will consider the intensities ratios
\begin{equation}
	r^i(t,\theta) = \frac{ \lambda^i(t,\vartheta)}{ \sum_{i'\in\mathbb I} \lambda^{i'}(t,\vartheta)}
	\label{eq:RatioModel}
\end{equation}
where $\theta=(\theta^1_1,\ldots,\theta^1_{\bar j}, \ldots, \theta^{\bar i}_1,\ldots,\theta^{\bar i}_{\bar j})$ denotes the new parameter vector.
Notation is justified by the following computation:
\begin{align}
	r^i(t,\theta) = & \frac{ \exp\left( \sum_{j\in\mathbb J} \vartheta^i_j X_ j(t)\right)}{ \sum_{i'\in\mathbb I} \exp\left( \sum_{j\in\mathbb J} \vartheta^{i'}_j X_ j(t)\right)}
	\nonumber \\
	= & \left[ \exp\left( - \sum_{j\in\mathbb J} \vartheta^i_j X_ j(t)\right)  \sum_{i'\in\mathbb I} \exp\left( \sum_{j\in\mathbb J} \vartheta^{i'}_j X_ j(t)\right)\right]^{-1}
	\nonumber \\
	= & \left[ \sum_{i'\in\mathbb I} \exp\left( \sum_{j\in\mathbb J} ( \vartheta^{i'}_j - \vartheta^i_j) X_ j(t)\right)\right]^{-1}
	\nonumber \\
	= & \left[ \sum_{i'\in\mathbb I} \exp\left( \sum_{j\in\mathbb J} ( \theta^{i'}_j - \theta^i_j) X_ j(t)\right)\right]^{-1}
	\nonumber \\
	= & \left[ 1 + \sum_{i'\in\mathbb I\setminus\{i\}} \exp\left( \sum_{j\in\mathbb J} ( \theta^{i'}_j - \theta^i_j) X_ j(t)\right)\right]^{-1}.
\label{equation:RatioWithRespectToTheta}
\end{align}
As hinted in the introduction, this framework is convenient to model a limit order book since daily movements are known to exhibit both intraday seasonality and sudden bursts of activities, sometimes in response to exogenous news, sometimes occurring for reasons less obvious to the outside observer. If we assume that such global market activity equally affects all types of orders, then such variations of market activity are taken into account by the non specified, possibly stochastic, baseline intensity $\lambda_0$. Therefore, we can estimate relative responses $\theta$ of order flows to specific covariates without having to account for the common intensity $\lambda_0$ due to global market activity. If different patterns of background intensities are observed in the data, then it is possible to incorporate these differences into covariates.

One should also note that by construction $\sum_{i\in\mathbb I} r^i(t,\theta)=1$. This gives another important feature of the model, which is that the intensities ratios $r^i$ are directly interpretable in terms of probability. Let us assume that the counting process $N^i$, $i\in\mathbb I$, are counting the number of events of (mutually exclusive) type $i$ occurring in the limit order book. Then the intensities ratio $r^i$ is the instantaneous probability that the next occurring event will be of type $i$. Our ratio model can thus estimate relative event probabilities independently of the variations of market activity that are assumed to be shared by the intensities of all processes. Several examples are provided in Section \ref{section:EmpiricalResults}.

\section{Likelihood analysis}
\label{section:LikelihoodAnalysis}
In this section, we show that the quasi-maximum likelihood estimator and the quasi-Bayesian estimator of the ratio model of Equation \eqref{eq:RatioModel} are consistent and asymptotically normal. Two close formulations are provided. In the first one, stationarity of covariates is assumed to obtain consistency and asymptotic normality. In the second one, this assumption is discarded, but it is assumed that the sample consists of repeated i.i.d. measurements to again obtain consistency and asymptotic normality. This second formulation is in agreement with the common practice in empirical finance to glue several trading days, or parts of trading days, into one single sample.

Let us introduce some notations used in the following analysis. For a tensor ${\sf T}=({\sf T}_{i_1,...,i_k})_{i_1,...,i_k}$, we write 
\begin{equation}
{\sf T}[u_1,...,u_k]
=
{\sf T}[u_1\otimes\cdots\otimes u_k]
=
\sum_{i_1,...,i_k}{\sf T}_{i_1,...,i_k}
u_1^{i_1}\cdots u_k^{i_k}
\end{equation}
for $u_1=(u_1^{i_1})_{i_1}$,..., $u_k=(u_k^{i_k})_{i_k}$. Brackets $[\ ,..., \ ]$ stand for a multilinear mapping. 
We denote by $u^{\otimes r}=u\otimes\cdots\otimes u$ the $r$ times tensor product of $u$. 
Let $\iota(\theta)=({\mathbf 0}_{\overline{j}},\theta)$, where ${\mathbf 0}_k=0\in\mathbb R^k$. We will write $\lambda^i(t,\theta)$ for $\lambda^i(t,\iota(\theta))$ for notational simplicity. 
Furthermore we write $\theta^i=(\theta^i_j)_{j\in\mathbb J}$ for $i\in\mathbb I$. Since $\theta^0=0$, we will use the notation $\mathbb I_0=\mathbb I\setminus\{0\}=\{1,\ldots,\bar i\}$ and consider a bounded domain $\Theta\subset \mathbb R^\mathsf{p}$, with $\mathsf{p}=\bar i \times \bar j$ as the parameter space of $\theta=(\theta^i_j)_{i\in\mathbb I_0,j\in\mathbb J}$. 

\subsection{Case of stationary covariates}

Let $T\in\mathbb R_+^*$. To estimate $\theta\in\Theta$ based on the observations on $[0,T]$, we consider the quasi-log likelihood (log partial likelihood) 
\begin{align}
\mathbb H_T(\theta) & =
\sum_{i\in\mathbb I}\int_0^T \log r^i(t,\theta)\,dN^i_t.
\label{eq:QuasiLogLik}
\end{align}
Obviously, the model is continuously extended to the closure $\overline{\Theta}$, and $\mathbb H_T$ is extended to there as a continuous function.
A quasi-maximum likelihood estimator (QMLE) is a measurable mapping $\hat{\theta}^M_T:\Omega\to\overline{\Theta}$ satisfying 
\begin{equation}
	\mathbb H_T(\hat{\theta}^M_T) = \max_{\theta\in\overline{\Theta}}\mathbb H_T(\theta)
	\label{eq:DefinitionQMLE}
\end{equation}
for all $\omega\in\Omega$. 
The quasi-Bayesian estimator (QBE) with respect to a prior density $\varpi$ is defined by 
\begin{equation}
	\hat{\theta}^B_T = \bigg[\int_\Theta\exp\big(\mathbb H_T(\theta)\big)\varpi(\theta)d\theta\bigg]^{-1}\int_\Theta\theta\exp\big(\mathbb H_T(\theta)\big)\varpi(\theta)d\theta. 
		\label{eq:DefinitionQBE}
\end{equation}
The QBE takes values in $\mathcal C[\Theta]$, the convex hull of $\Theta$. 
We assume $\varpi$ is continuous and satisfies $0<\inf_{\theta\in\Theta}\varpi(\theta)\leq\sup_{\theta\in\Theta}\varpi(\theta)<\infty$. 
These estimators are called together quasi-likelihood estimators. We now investigate asymptotic properties of the estimators when $T\to\infty$ in two cases.

Let $\mathbb X=(X_j)_{j\in\mathbb J}$. 
In this first step, in order to simplify the statements, we assume that 
the process $(\lambda_0,\mathbb X)$ is stationary. 
Denote by $\theta^*=((\vartheta^*)^i_j-(\vartheta^*)^0_j)_{i\in\mathbb I_0,j\in\mathbb J}$ the true value of $\theta$, 
and assume that $\theta^*\in\Theta$. 
Denote by $\mathcal B_I$ the $\sigma$-field generated by 
$\{\lambda_0(t),\mathbb X(t);\>t\in I\}$ for $I\subset\mathbb R_+$. 
Let 
\begin{equation} 
\alpha(h) = \sup_{t\in\mathbb R_+}\sup_{A\in\mathcal B_{[0,t]},B\in\mathcal B_{[t+h,\infty)}}
\big|P[A\cap B]-P[A]P[B]\big|
\end{equation}
for $h>0$. 
We consider the following conditions. 
\begin{description}
\item{\textbf{[A1]}}
The process $(\lambda_0,\mathbb X)$ is stationary, 
$\lambda_0(0)\in L_{\infty-}=\cap_{p>1}L_p$ and 
$\exp(|X_j(0)|)\in L_{\infty-}$ for all $j\in\mathbb J$. 
\item{\textbf{[A2]}} The function $\alpha$ is rapidly decreasing, that is, 
$\limsup_{h\to\infty}h^L\alpha(h)<\infty$ for every $L>0$. 
\end{description}
Let  
\begin{equation} 
\rho^i(x,\theta) = 
\left[\sum_{i'\in\mathbb I}\exp\bigg(x\big[\theta^{i'}-\theta^i\big]\bigg)\right]^{-1}
\end{equation}
for $x\in\mathbb R^{\overline{j}}$ and 
$\theta=(\theta^i)_{i\in\mathbb I_0}\in\mathbb R^\mathsf p$. 
Let 
\begin{equation} 
\Lambda(w,x) =
w\sum_{i\in\mathbb I}\exp\big(x\big[\vartheta^{*i}\big]\big)
\end{equation}
for $w\in\mathbb R_+$ and $x\in\mathbb R^{\overline{j}}$. 
Denote by ${\sf V}(x,\theta)$ 
the variance matrix of the $(1+\overline{i})$-dimensional multinomial distribution 
$\text{M}(1;\pi_0,\pi_1,...,\pi_{\overline{i}})$ with 
$\pi_i=\rho^i(x,\theta)$, $i\in\mathbb I$, 
and ${\sf V}(x,\theta)_{\alpha,\alpha'}$ is ${\sf V}(x,\theta)$'s 
$(\alpha+1,\alpha'+1)$-element. 
Let ${\sf V}_0(x,\theta)=({\sf V}(x,\theta)_{\alpha,\alpha'})_{\alpha,\alpha'\in\mathbb I_0}$.  
Define a symmetric tensor $\Gamma$ by 
\begin{equation} 
\Gamma[u^{\otimes2}]
=
E\bigg[\bigg({\sf V}_0(\mathbb X(0))\otimes \mathbb X(0)^{\otimes2}\bigg)[u^{\otimes2}]
\Lambda(\lambda_0(0),\mathbb X(0))\bigg]
\label{equation:GammaThm1}
\end{equation}
for $u\in\mathbb R^\mathsf p$, 
where 
${\sf V}_0(x)={\sf V}_0(x,\theta^*)$. 
The matrix $\Gamma$ is nonnegative definite. 
We assume 
\begin{description}
	\item{\textbf{[A3]}} $\det \Gamma>0$. 
\end{description}
The non-degeneracy of $\Gamma$ is a local condition. 
However, the global identifiability condition follows from this condition 
in the present model, as seen later. 

Let $\theta^*\in\Theta$ denote the true value of $\theta$. 
Denote by $C_p(\mathbb R^\mathsf p)$ the space of continuous functions on $\mathbb R^\mathsf p$ of at most polynomial growth. 
We write 
$\hat{u}^M_T=\sqrt{T}\big(\hat{\theta}^M_T-\theta^*\big)$ and 
$\hat{u}^B_T=\sqrt{T}\big(\hat{\theta}^B_T-\theta^*\big)$. 
Denote by $\zeta$ a $\mathsf p$-dimensional standard Gaussian vector. 
The quasi-likelihood analysis ensures convergence of moments 
as well as asymptotic normality of the quasi-likelihood estimators. 

\begin{theorem}
Suppose that $[A1]$, $[A2]$ and $[A3]$ are satisfied. Then 
\begin{equation}\label{170904-1}
E\big[f(\hat{u}^{\sf A}_T)\big] \to E\big[f(\Gamma^{-1/2}\zeta)\big]\quad 
\end{equation}
for any $f\in C_p(\mathbb R^\mathsf p)$ and ${\sf A}\in \{M,B\}$.
\label{170925-2}
\label{thm:Thm1}
\end{theorem}

Proof is given in the appendix.

\begin{example}
\label{example:ExampleHawkesThm31}
Let $H$ be a Hawkes process with exponential kernel with parameters $(\mu, \alpha, \beta)\in(\mathbb R_+^*)^3$. For the sake of illustrating Theorem \ref{thm:Thm1}, we consider a version of model \eqref{eq:CoxModel} with $\bar i=1$, $\bar j=1$, and in which the common baseline intensity $\lambda_0$ is the intensity of the process $H$. We thus have:
\begin{equation}
	\lambda^i(t,\vartheta) = \left[ \mu + \int_0^t \alpha e^{-\beta(t-s)}\,dH(s) \right] \exp\left( \vartheta^i_1 X_1(t) \right),
	\quad i=0,1,
\end{equation}
where $X_1$ is a Markov chain with states $\{-1,1\}$ and infinitesimal generator $\left(\begin{matrix} \lambda_X & -\lambda_X \\ -\lambda_X & \lambda_X \end{matrix}\right)$.
Then we obviously have $\mathsf p=1$, i.e. $\theta=\vartheta^1_1-\vartheta^0_1$ is the single parameter to be fit.
In this specific case, one may explicitly compute the asymptotic variance of Equation \eqref{equation:GammaThm1}:
\begin{equation}
\Gamma = \frac{\mu}{1-\frac{\alpha}{\beta}} \frac{e^{\theta^*}}{(1+e^{\theta^*})^2} \left[ \cosh((\vartheta^*)^1_0) + \cosh((\vartheta^*)^1_1) \right] \in\mathbb R_+^*.
\end{equation}
We run 1000 simulations of the model with Hawkes parameters $\mu=0.5, \alpha=1, \beta=2$, covariate parameter $\lambda_X=0.5$, and horizon $T=1000$. True values of the parameters are $(\vartheta^*)^1_1=0.75$, $(\vartheta^*)^0_1=-0.75$ so that $\theta^*=1.5$. For each simulation we estimate $\theta$. The empirical density of the estimated values of $\theta-\theta^*$ is plotted on Figure \ref{figure:NumericalExperiment1Thm31} in dots and full line. On the same plot is provided in dashed lines a Gaussian distribution fitted on these estimated values, as well as the theoretical Gaussian distribution given by Theorem \ref{thm:Thm1}, i.e. with the asymptotic variance given at Equation \eqref{equation:GammaThm1}. Experimental values agree with the results of Theorem \ref{thm:Thm1}.
\begin{figure}
\begin{center}
	\includegraphics[scale=0.45, page=2]{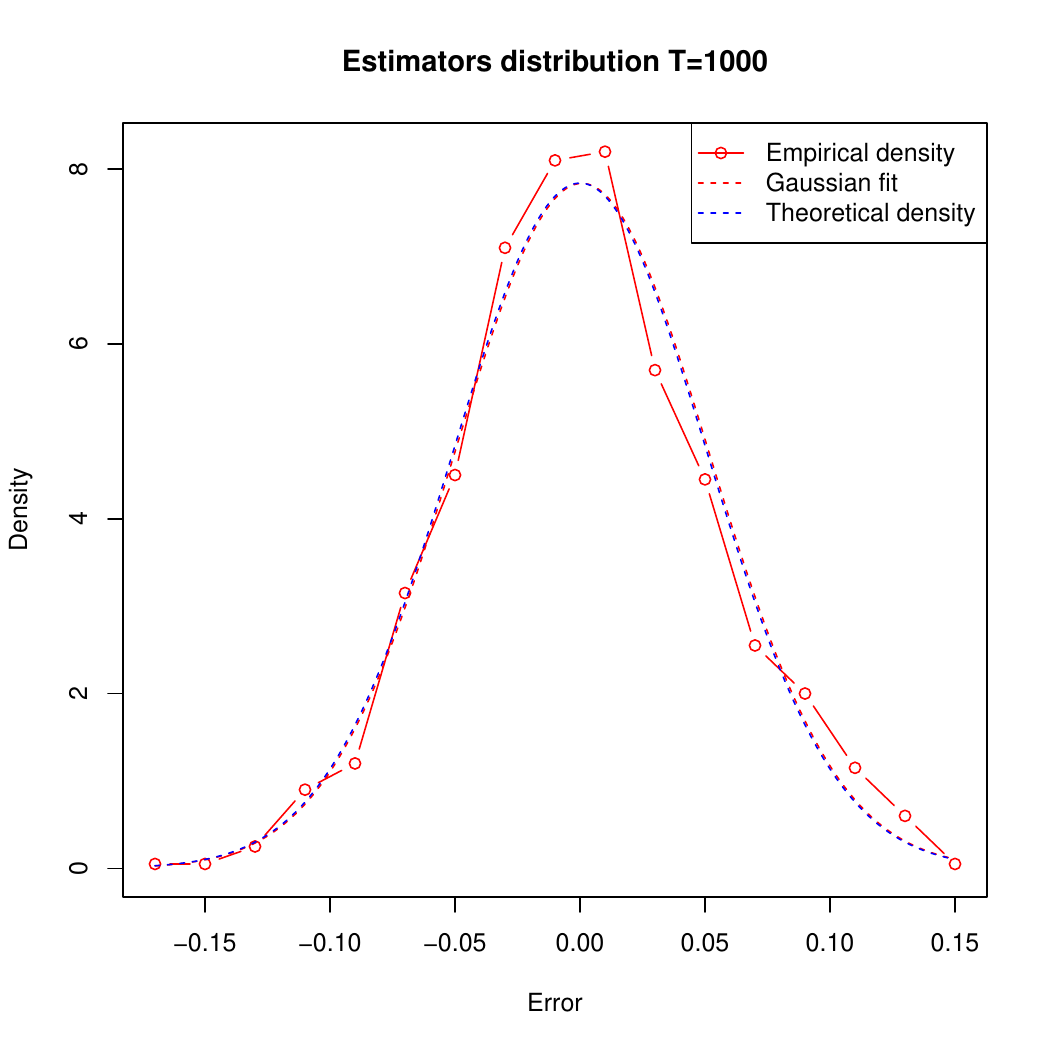}
	\caption{Distribution of the estimator of the estimation error $\theta-\theta^*$ for the model defined in Example \ref{example:ExampleHawkesThm31}. Asymptotic variance defined at Equation \eqref{equation:GammaThm1} is retrieved experimentally.}
	\label{figure:NumericalExperiment1Thm31}
\end{center}
\end{figure}
\end{example}

\begin{example}
This paper does not focus on the baseline intensity. However, in the case of a parametric model where the baseline intensity is fully specified, then the ratio approach is particularly helpful as it reduces the dimension of the space of the parameters, and can thus help improving numerical estimations by choosing appropriate starting points for the optimization routines. As an illustration, let us consider a model similar to the previous example, with a Hawkes process $H$ with $\bar k=2$ exponential kernels as baseline intensity, $\bar i=2$, i.e. 3 processes, and $\bar j=2$ covariates (same Markov chains as in the previous example):
\begin{equation}
	\lambda^i(t,\vartheta) = \left[ 1 + \sum_{k=1}^{\bar k}\int_0^t \alpha_k e^{-\beta_k(t-s)}\,dH(s) \right] \exp\left( \sum_{j=1}^{\bar j}\vartheta^i_j X_j \right).
\label{eq:Example2-Intensity}
\end{equation}
Maximum likelihood estimators of this model are directly computable when the process $H$ is observable. This requires an optimization on 10 parameters. As an alternative estimation procedure, we can estimate the ratio model (4 parameters $\theta^i_j$, $i=1,2$, $j=1,2$) and then estimate the full model likelihood with the constraints that the differences $\vartheta^i_j-\vartheta^0_j$ are kept equal to the estimated values $\theta^i_j$ (dimension 6). All 10 parameters are thus estimated. The estimated values can finally be used as a starting point for a final maximization of the likelihood of the full model, in order to ensure that we keep the desired properties of the maximum likelihood estimators.
Estimation results on simulations of the model \eqref{eq:Example2-Intensity} are presented in Table \ref{table:Example2-Results}. 
\begin{table}[htbp]
\centering
\footnotesize
\begin{tabular}{|c|c|c|c|c|c|c|c|c|c|c|}
\hline
 & $\alpha_0$ & $\alpha_1$ & $\beta_0$ & $\beta_1$ & $\vartheta^0_0$ & $\vartheta^0_1$ & $\vartheta^1_0$ & $\vartheta^1_1$ & $\vartheta^2_0$ & $\vartheta^2_1$
\\ \hline
True & 1.000 & 2.000 & 2.000 & 10.000 & 0.500 & 1.000 & 0.500 & -1.000 & -0.500 & 1.000 
\\ \hline
\multirow{2}{*}{Full} & 0.942 & 1.507 & 2.131 & 5.647 & 0.499 & 1.011 & 0.517 & -1.003 & -0.504 & 1.015 
\\
 & (0.819) & (1.150) & (0.878) & (2.492) & (0.013) & (0.013) & (0.049) & (0.020) & (0.004) & (0.019) 
\\ \hline
\multirow{2}{*}{Combined} & 0.994 & 1.993 & 1.981 & 9.923 & 0.500 & 1.000 & 0.500 & -1.000 & -0.500 & 1.000 
\\
  & (0.155) & (0.154) & (0.233) & (1.034) & (0.005) & (0.005) & (0.004) & (0.005) & (0.004) & (0.005) 
\\ \hline
\end{tabular}
\caption{Numerical results for the estimation of the 10 parameters of the model defined at Equation \eqref{eq:Example2-Intensity}. Standard deviations a given in parenthesis. See text for details.}
\label{table:Example2-Results}
\end{table}
\normalsize
``Full'' denotes the standard maximum likelihood estimation. ``Combined'' denotes the mixed method involving a ratio estimation as a first step. Optimizations are carried out with a Nelder-Mead algorithm (Python scipy implementation) with random starting points (using a standard Gaussian distribution). With these parameters and an horizon $T=10\,000$, samples have roughly $33\,000$ data points in average. Out of 197 tests, the ``Full' estimation converged only 10 times, while the ``Combined'' method converged 192 times out of 202. Estimates produced by the ``Combined'' method are closer to the true values and have a smaller empirical standard deviations by a factor 2 to 10 (except for $\vartheta^2_0$, already well-estimated in the ``Full'' case). Strong improvements are observed in the Hawkes parameters, where the standard ``Full'' method struggles to produce close estimates, while the ``Combined'' method starting with a ratio estimation yields much more accurate results.
\end{example}

\subsection{Repeated measurements}

We complete the previous analysis with a formulation of our estimation result that may be convenient in finance and in other fields of applications. We consider a sequence of observations of intraday data. 
The observations may be non-ergodic each day. 
We shall consider intervals $I^{(k)}=[O_k,C_k]$ ($k\in\mathbb N$) of the same length such that $0\leq O_1<C_1\leq O_2<C_2\leq\cdots$. We consider the counting processes $N^i=(N^i_t)_{t\in\mathbb R_+}$ of the previous section. However, in this section, the observations are $((N^i_t)_{i\in\mathbb I},\mathbb X(t))_{t\in I^{(k)}}$, $k=1,...,T$. 
As before, it is assumed that the point processes $N^i$ $(i\in\mathbb I)$ have no common jumps. The stationarity of each process $\big(\lambda^i(t,\vartheta^*)\big)_{t\in I^{(k)}}$ is not assumed here. 

We are interested in estimation of the parameter $\theta=(\theta^i_j)_{i\in\mathbb I_0,\>j\in\mathbb J}$ defined earlier by Equation \eqref{eq:DefinitionTheta}.
The estimation will be based on the random field $\mathbb H_T$ re-defined by 
\begin{equation}\label{eq:RepeatedQuasiLogLik} 
\mathbb H_T(\theta) = \sum_{k=1}^T\sum_{i\in\mathbb I}\int_{I^{(k)}}\log r^i(t,\theta)dN^i_t.
\end{equation}
Then the QMLE and the QBE are defined by Equations \eqref{eq:DefinitionQMLE} and \eqref{eq:DefinitionQBE}, respectively, but for $\mathbb H_T$ given by Equation \eqref{eq:RepeatedQuasiLogLik}. 

Let $\mathcal X_k=\big(\lambda_0(t),\mathbb X(t)\big)_{t\in I^{(k)}}$. Let $\mathcal G_k=\sigma[\mathcal X_1,...,\mathcal X_k]$ and let $\mathcal H_k=\sigma[\mathcal X_k,\mathcal X_{k+1},...]$. The $\alpha$-mixing coefficient for $\mathcal X=(\mathcal X_k)_{k\in\mathbb N}$ is defined by 
\begin{equation}
\alpha^\mathcal X(h) = \sup_{k\in \mathbb N} \sup_{A\in\mathcal G_k,\>B\in\mathcal H_{k+h}}\big|P[A\cap B]-P[A]P[B]\big|. 
\end{equation}

Let us then define the new conditions.
\begin{description}
	\item{\textbf{[C1]}} The sequence $(\mathcal X_k)_{k\in\mathbb N}$ is identically distributed. Moreover, $\sup_{t\in I^{(1)}}\|\lambda_0(t)\|_p<\infty$ and $\max_{j\in\mathbb J}\sup_{t\in I^{(1)}}\|\exp(p|X_j(t)|)\|_1<\infty$ for all $p>1$. 

	\item{\textbf{[C2]}} For every $L>0$, $\limsup_{h\to\infty}h^L\alpha^\mathcal X(h)<\infty$. 
\end{description}

Define the symmetric tensor $\Gamma$ by 
\begin{equation}
\Gamma[u^{\otimes2}] = E\bigg[\int_{I^{(1)}}\bigg({\sf V}_0(\mathbb X(t))\otimes \mathbb X(t)^{\otimes2}\bigg)[u^{\otimes2}] \Lambda(\lambda_0(t),\mathbb X(t))dt\bigg]
\end{equation} 
for $u\in\mathbb R^\mathsf p$. The matrix $\Gamma$ is nonnegative definite. More strongly we assume 
\begin{description}
\item{\textbf{[C3]}} $\det \Gamma>0$. 
\end{description}

In a way similar to Theorem \ref{thm:Thm1}, it is possible to prove the following theorem. 
\begin{theorem} 
Suppose that $[C1]$, $[C2]$ and $[C3]$ are satisfied. Then 
\begin{equation}
E\big[f(\hat{u}^{\sf A}_T)\big] \to E\big[f(\Gamma^{-1/2}\zeta)\big]\quad 
\end{equation}
as $T\to\infty$ 
for any $f\in C_p(\mathbb R^\mathsf p)$ and ${\sf A}\in \{M,B\}$. 
\end{theorem}
The proof is omitted.

\section{Information criteria and penalization}
\label{section:Penalization}

Since the ratio model flexibly incorporates various covariates processes and since it may have a large number of parameters, we need information criteria for model selection and other regularization methods for sparse estimation. Though the inference of the ratio model is based on the quasi-likelihood analysis, we can still apply information criterion like CAIC (consistent AIC) and BIC by using $\mathbb{H}_T$. See \cite{bozdogan1987model} for exposition of information criteria for model selection. 

Let $a_T$ be a sequence of positive numbers such that $a_T\to\infty$ and $a_T/T\to0$ as $T\to\infty$. Let $\mathbb{K}\subset\mathbb{I}_0\times\mathbb{J}$. We consider a sub-model $S_\mathbb{K}$ of $\Theta$ such that 
\begin{equation} 
S_\mathbb{K} = \{\theta\in\Theta;\>\theta^i_j=0\ ((i,j)\in\mathbb{K}^c)\}. 
\end{equation}
Let 
\begin{equation} 
C_T(S_\mathbb{K}) = -2\mathbb{H}_T(\hat{\theta}_\mathbb{K})+\mathsf{d}(S_\mathbb{K})a_T
\end{equation}
where $\hat{\theta}_\mathbb{K}$ (depending on $T$) is denoting the QMLE or QBE in the sub-model $S_\mathbb{K}$ and $\mathsf{d}(S_\mathbb{K})$ is the dimension of $S_\mathbb{K}$. More precisely, the QMLE $\hat{\theta}^M_{\mathbb K}$ is defined as an estimator that satisfies 
\begin{equation}
	\mathbb H_T(\hat{\theta}^M_{\mathbb K}) = \max_{\theta\in \overline{S_\mathbb K}}\mathbb H_T(\theta),
\end{equation}
and the QBE $\hat{\theta}^B_{\mathbb K}$ is defined by 
\begin{equation}
\hat{\theta}^B_{\mathbb K} = \bigg[\int_{S_\mathbb K}\exp\big(\mathbb H_T(\theta)\big)\varpi_{S_\mathbb K}(\theta)\bigg]^{-1} \int_{S_\mathbb K}\theta\exp\big(\mathbb H_T(\theta)\big)\varpi_{S_\mathbb K}(\theta)d\theta
\end{equation}
for a continuous prior density $\varpi_{S_\mathbb K}$ on $S_\mathbb K$ satisfying $0<\inf_{\theta\in S_\mathbb K}\varpi_{S_\mathbb K}(\theta)\leq
\sup_{\theta\in S_\mathbb K}\varpi_{S_\mathbb K}(\theta)<\infty$. We denote by $S_{\mathbb{K}^*}$ the minimum model that includes $\theta^*$, in other words, $(\theta^*)^i_j\not=0$ if and only if $(i,j)\in\mathbb{K}^*$. The QMLE and QBE for $\theta$ restricted to the sub-model $S_{\mathbb K^*}$ 
are generically denoted by $\hat{\theta}_{\mathbb K^*}$. 
In particular, 
\begin{equation}
C_T(S_{\mathbb K^*}) = -2\mathbb H_T(\hat{\theta}_{\mathbb K^*})+\mathsf d(S_{\mathbb K^*})a_T.
\end{equation}

\begin{proposition}\label{20180504-1}
Suppose that $[A1]$, $[A2]$ and $[A3]$ are satisfied. 
Then
\begin{description}
\item[(i)] If $\theta^*\not\in S_\mathbb{K}$, then 
\begin{equation}
C_T(S_\mathbb{K})-C_T(S_{\mathbb{K}^*}) \to^p \infty \qquad(T\to\infty)
\end{equation}
\item[(ii)] If $\theta^*\in S_\mathbb{K}$ and $S_\mathbb{K}\not=S_{\mathbb{K}^*}$, then 
\begin{equation}
C_T(S_\mathbb{K})-C_T(S_{\mathbb{K}^*}) \to^p \infty \qquad(T\to\infty)
\end{equation}
\end{description}
\end{proposition}

A proof of Proposition \ref{20180504-1} is given in Appendix for selfcontainedness. According to Proposition \ref{20180504-1}, we should select a model $S_{\widehat{\mathbb{K}}}$ that attains $\min\{C_T(S_\mathbb{K});\>\mathbb{K}\subset\mathbb{I}_0\times\mathbb{J}\}$. Then the selection consistency holds as 
\begin{equation} 
P\big[\widehat{\mathbb{K}}=\mathbb{K}^*\big] \to 1\qquad(T\to\infty).
\end{equation}
Based on the quasi-likelihood function $\mathbb{H}_T$, the criterion $C_T(S_\mathbb{K})$ gives the quasi-consistent AIC (QCAIC) when $a_T=\log T+1$, and the quasi-BIC (QBIC) when $a_T=\log T$ among 
many other possible choices of $a_T$. A Hannan and Quinn type of quasi-information criterion (QHQ) is the case where $a_T=c\log\log T$ for $c>2$. The quasi-AIC (QAIC) is the case where $a_T=2$ but it cannot exclude overfitting though it excludes misspecified models, as usual and as suggested in the proof of Proposition \ref{20180504-1}. 

Recently penalized quasi-likelihood analysis for sparse estimation has been developed. We consider a penalty function $p_\lambda$ such that $p_\lambda(x)=p_\lambda(-x)$, $p_\lambda(0)=0$, $p_\lambda$ is non-decreasing on $x\geq0$, $p_\lambda$ is differentiable except for $x=0$, and $\lim_{x\downarrow 0}x^{-q}p_\lambda(x)=\lambda>0$ for some $q\in(0,1]$. This class of penalty functions includes the ones of LASSO (\cite{tibshirani1996regression}), Bridge (\cite{frank1993statistical}) and SCAD (\cite{fan2001variable}). Suppose that we have a QLA based on the random field $\mathbb{H}_T$. Then we have the penalized contrast function 
\begin{equation}\label{20180504-10} 
-\mathbb{H}^\dagger_T(\theta) = -\mathbb{H}_T(\theta)+\sqrt{T}\sum_{i\in\mathbb{I}_0,j\in\mathbb{J}} p_\lambda(\theta^i_j)
\end{equation}
and the penalized QMLE $\hat{\theta}_\lambda$ (depending on $T$) is associated by
\begin{equation}\label{20180504-11} 
\hat{\theta}_\lambda \in \text{argmin}_{\theta\in\overline{\Theta}} \big\{-\mathbb{H}^\dagger_T(\theta)\big\}.
\end{equation}
In the case $q=1$, the polynomial type large deviation inequality is inherited from $\mathbb{H}_T$ to $\mathbb{H}_T^\dagger$, as a result, we obtain asymptotic properties of $\hat{\theta}_\lambda$ as well as $L^p$-boundedness of the error. Asymptotic distribution becomes slightly involved. In a similar way, in the case $q<1$, it is possible to derive asymptotic properties of $\hat{\theta}_\lambda$. A prominent property is selection consistency. Thanks to the QLA theory having a polynomial type large deviation inequality, we can prove that the probability of correct model selection is $1-O(T^{-L})$ as $T\to\infty$ for every $L>0$. See \cite{kinoshita-yoshida2018} for details. An advantage of working with the QLA theory is that asymptotic properties of regularization methods are obtained in a unified manner without using any specific nature of the stochastic process. The QLA was used in \cite{umezu2015aic} for a generalized linear model. 

Another approach is toward the adaptive LASSO (\cite{zou2006adaptive}) and the least square approximation method (\cite{wang2007unified}). \cite{de2012adaptive} took this approach to sparse estimation of ergodic diffusions. Since we can assume strong mode of convergence for the initial estimator constructed within the QLA framework, we can obtain selection consistency with error probability that is of $O(T^{-L})$ for any $L>0$ as well as the oracle properties of the penalized estimator. For details, see \cite{suzuki-yoshida2018}. 

Related to regularization methods for point processes, \cite{fan2002variable} extended the nonconcave penalized likelihood approach to the Cox proportional hazards model. \cite{yue2015variable} treated variable selection in spatial point processes. \cite{hansen2015lasso} derived probabilistic inequalities for multivariate point processes in the context of LASSO.

\section{Empirical results - Dependencies analysis}
\label{section:EmpiricalResults}

This section describes the high-frequency financial data (Section \ref{subsec:Data}) and several empirical studies conducted with it. A detailed example of financial analysis allowed by the ratio model is provided in Section \ref{subsec:MBMAImbalance}, where the determination of the sign of the next trade in a limit order book is carried out with several covariates. Other examples are provided in subsequent sections, illustrating the influence of the spread and queue sizes in the order book.

\subsection{Data}
\label{subsec:Data}

We use Reuters TRTH tick by tick data for 36 stocks traded on the Paris stock Exchange on most of the year 2015. The list of the stocks is given in Appendix \ref{sec:ricList}. For each stock and each trading day, the orders flow is reconstructed using the method described in \cite{MuniToke2017}, and we keep all trades and quotes recorded between 9:30 am and 5:00 pm\footnote{Removing the beginning and the trading day has been done as a usual precautionary step when dealing with high-frequency trades and quotes databases, as one may sometimes be concerned with data quality in very busy periods. However, this precaution may actually not be necessary. For example, when recomputing Figure \ref{figure:ProbabilityImbalanceLastSpread} below using the whole trading day, no we observe no significant visual difference with the version presented in this paper.}. 
When adding all stocks and trading days, the sample represents more than one billion events (more than $50$ millions market orders, more than $500$ millions limit orders, and more than $450$ millions cancellations).

All models are numerically estimated using \verb+R+ by quasi-likelihood maximization. For practical purposes, we provide practical explicit expressions for the log partial likelihood. Using Equation \eqref{equation:RatioWithRespectToTheta}, one may write Equation \eqref{eq:QuasiLogLik} with respect to $\theta$ as:
\begin{align}
	\mathbb H_T(\theta) = & - \sum_{i\in\mathbb I} \int_0^T \log \left[ \sum_{i'\in\mathbb I} \exp\left( \sum_{j\in\mathbb J} ( \theta^{i'}_j - \theta^i_j) X_ j(t)\right)\right] \,dN^i(t)
	\nonumber \\
	= & - \sum_{i\in\mathbb I} \int_0^T \log \left[ 1 + \sum_{i'\in\mathbb I\setminus\{i\}} \exp\left( \sum_{j\in\mathbb J} ( \theta^{i'}_j - \theta^i_j) X_ j(t)\right)\right] \,dN^i(t)
\label{equation:GeneralQuasiLogLikelihood}
\end{align}
Equation \eqref{eq:RepeatedQuasiLogLik} is similarly written. In the example case $\mathbb I=\{0,1,2\}$ with three counting processes, the quasi-log likelihood given at Equation \eqref{equation:GeneralQuasiLogLikelihood} is written:
\begin{align}
	\mathbb H_T(\theta) = & - \int_0^T \log \left[ 1 + \exp\left( \sum_{j\in\mathbb J} \theta^{1}_j X_ j(t)\right) + \exp\left( \sum_{j\in\mathbb J} \theta^{2}_j X_ j(t)\right) \right] \,dN^0(t) 
	\nonumber \\ & 
	- \int_0^T \log \left[ 1 +  \exp\left( - \sum_{j\in\mathbb J} \theta^1_j X_ j(t)\right) + \exp\left( \sum_{j\in\mathbb J} (\theta^2_j-\theta^1_j) X_ j(t)\right) \right] \,dN^1(t)
	\nonumber \\ & 
	- \int_0^T \log \left[ 1 +  \exp\left( - \sum_{j\in\mathbb J} \theta^2_j X_ j(t) \right) + \exp\left( \sum_{j\in\mathbb J} (\theta^1_j-\theta^2_j) X_ j(t)\right) \right] \,dN^2(t)
\end{align}

An important feature of the model may be underlined here. Many mathematical models of limit order books, and models in high-frequency finance in general, rely on simple point processes, i.e. processes for which one may not observe two simultaneous events. Confronting such modeling to empirical data thus often requires that the timestamps of all recorded events are different. This condition, even with resolutions down to milliseconds or even microseconds, is not verified on financial markets, where burst of activities are translated into data files by multiple events at the same timestamp. Randomization, i.e. adding a small random quantity to the observed timestamps, is often proposed as a patch to force all timestamps to be different. 
A nice feature of the setting proposed here is that such patches are not necessary : likelihood \eqref{equation:GeneralQuasiLogLikelihood} does not depend on interval times between consecutive events, but only on the number of events. Therefore, the model can be applied even on data with low resolution timestamps.

\subsection{Trades signs in a limit order book}
\label{subsec:MBMAImbalance}

It is well known that imbalance is good proxy to the sign of the next trade. Empirical observations can be found in e.g. \cite{Lipton2013}. \cite{Lehalle2017}, among others, is an attempt to incorporate this signal in trading strategies.
Let us define the imbalance in the order book observed at time $t$ by:
\begin{equation}
	i(t) = \frac{q_1^B(t)-q_1^A(t)}{q_1^B(t)+q_1^A(t)},
	\label{equation:imbalanceDefinition}
\end{equation}
where $q_1^A(t)$ and $q_1^B(t)$ are the number of shares available at time $t$ at the best quote on the ask side and bid side respectively. An imbalance close to $+1$ indicates a very small available volume on the ask side and consequently a probable upward price move, while an imbalance close to $-1$ increases the probability of a downward price move.
Using our model, we can infer how the imbalance affects the relative intensities of bid and ask market orders.
Let us assume that counting processes of ask market orders $N^{MA}$ and bid market orders $N^{MB}$ are point processes with intensities :
\begin{align}
	\lambda^{MA}(t,\vartheta^{MA}) = \lambda_0(t) \exp\left[\vartheta^{MA}_0+\vartheta^{MA}_1 i(t)\right],
	\nonumber \\ 
	\lambda^{MB}(t,\vartheta^{MB}) = \lambda_0(t) \exp\left[\vartheta^{MB}_0+\vartheta^{MB}_1 i(t)\right],
	\label{equation:Model-MBMA-Imbalance}
\end{align}
where $\lambda_0(t)$ is a potentially random baseline intensity.
Parameters $\theta_i=\theta^{MA}_i-\theta^{MB}_i$, $i=0,1$ can be straightforwardly estimated by maximization of the quasi-likelihood. As mentioned above, a very interesting feature of our intensity model is that the intensities ratios $r_{MA}$ and $r_{MB}$ are then directly interpretable as the instantaneous probability that the next market order will occur on the ask side and on the bid side respectively.

The model is fitted monthly for all stocks, from January 2015 to November 2015, which, excluding gaps in the database, gives 390 fits of the model. On Figure \ref{figure:ProbabilityImbalance}, we plot for two samples the empirical probability that for a given imbalance, the next trade occurs on the ask side, and the numerical estimation of the theoretical probability $r_{MA}(i)=\left[1+\exp(-\theta_0-\theta_1 i)\right]^{-1}$ for a level $i$ of imbalance given by the fit of our ratio model. For representativity we rank the 390 fits by increasing mean $L^2$ distance between empirical and fitted probabilities, and we select two samples representing the $20\%$ and $80\%$ quantiles of this error distribution. All 390 fits are available upon request.
\begin{figure}
\begin{center}
\begin{tabular}{cc}
\includegraphics[scale=0.42, page=2]{{Cox-MBMA-Imbalance-MONTHLY-SELECTED}.pdf}
&
\includegraphics[scale=0.42, page=1]{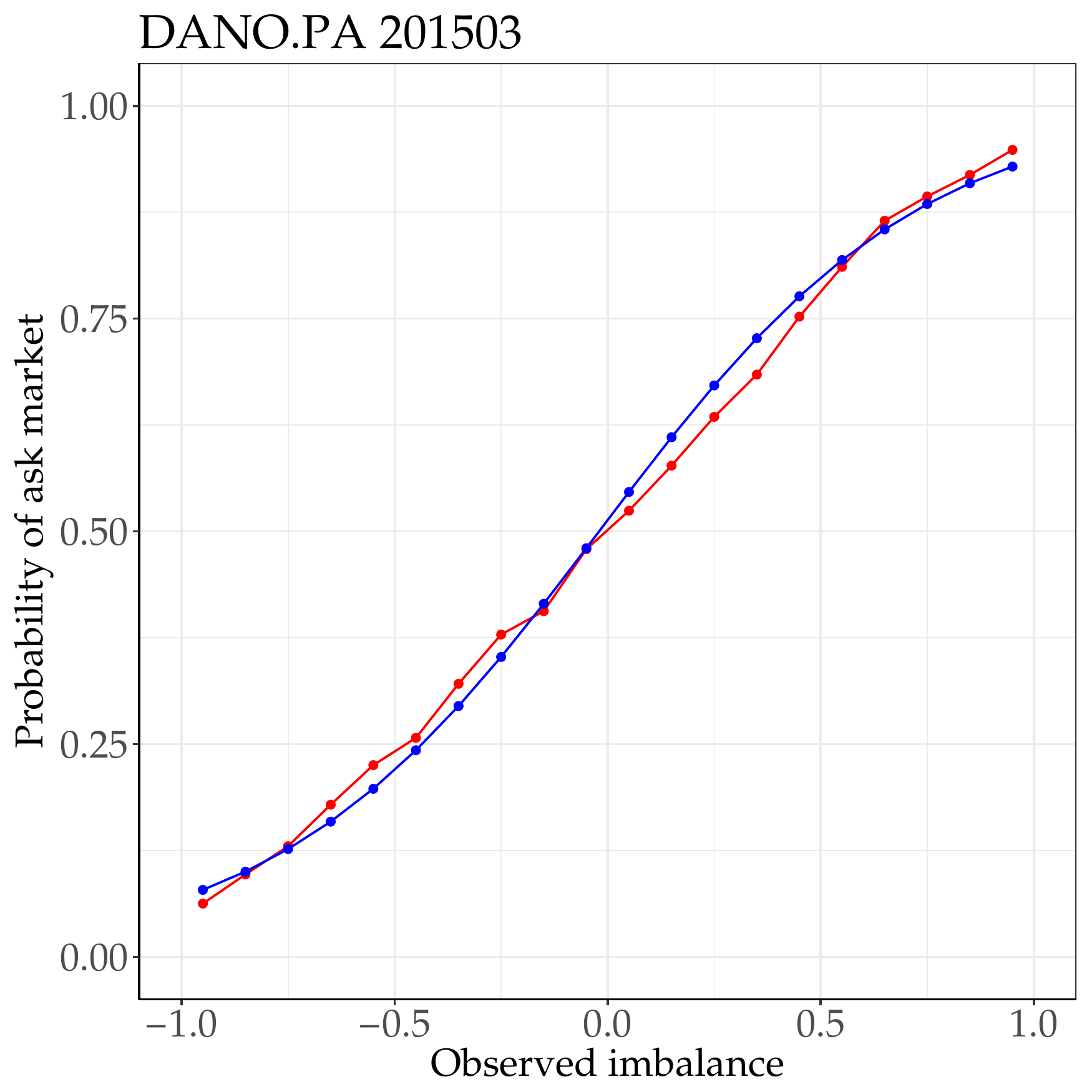}
\end{tabular}
\caption{Empirical probability (in red) and fitted probability (in blue) that the next market order is an ask market order, given the observed imbalance. Selected stocks show the $20\%$ (left) and $80\%$(right) quantiles measured in mean $L^2$ error.}
\label{figure:ProbabilityImbalance}
\end{center}
\end{figure}

The simple ratio model with one covariate provides a very good fit of the probability of the next trade. 
But we can obviously investigate further possible factors influencing the sign of the next trade. It has been shown that times series of signs of trades ($\epsilon(t)=+1$ for an ask trade at time $t$, $-1$ for a buy trade) exhibit slowly decreasing autocorrelation functions (see e.g., \cite{Lillo2004,Bouchaud2004}). We can include the sign of the last trade in the ratio model, and assume that counting processes of ask market orders $N^{MA}$ and bid market orders $N^{MB}$ are point processes with intensities :
\begin{align}
	\lambda^{MA}(t,\vartheta^{MA}) = \lambda_0(t) \exp\left[\vartheta^{MA}_0+\vartheta^{MA}_1 i(t) + \vartheta^{MA}_2 \epsilon(t)\right],
	\nonumber \\ 
	\lambda^{MB}(t,\vartheta^{MB}) = \lambda_0(t) \exp\left[\vartheta^{MB}_0+\vartheta^{MB}_1 i(t) + \vartheta^{MB}_2 \epsilon(t)\right],
	\label{equation:Model-MBMA-ImbalanceLast}
\end{align}
Figure \ref{figure:ProbabilityImbalanceLast} plots the updated results for two samples. For representativity, we again select the two samples representing the $20\%$ and $80\%$ quantile measured in mean $L^2$ error. All 390 fits are available upon request.
\begin{figure}
\begin{center}
\begin{tabular}{cc}
\includegraphics[scale=0.42, page=2]{{Cox-MBMA-ImbalanceLastmarket-MONTHLY-SELECTED}.pdf}
&
\includegraphics[scale=0.42, page=1]{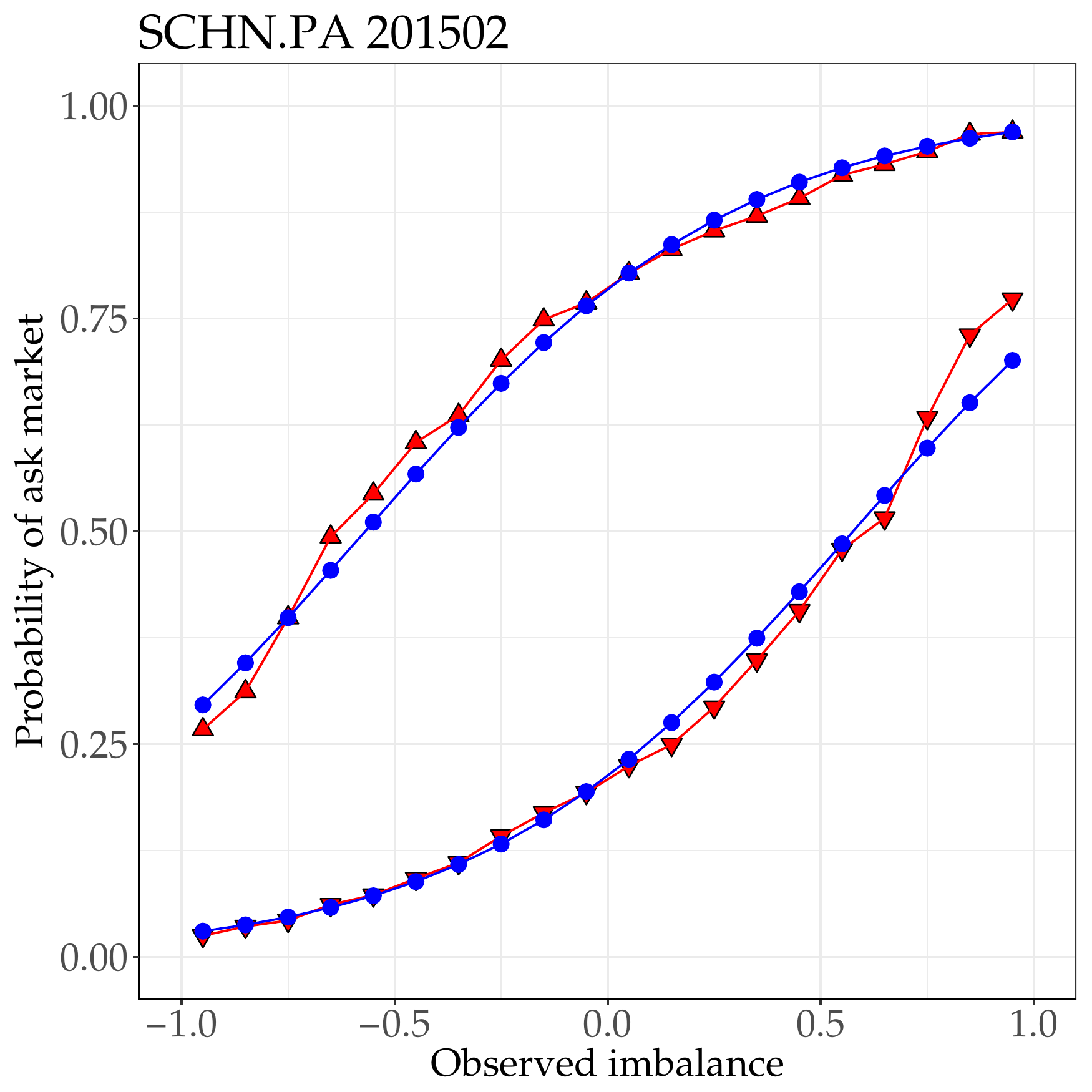}
\end{tabular}
\caption{Empirical probability (in red) and fitted probability (in blue) that the next market order is an ask market order, given the observed imbalance. Upward (resp. downward) triangles indicate that the last trade was an ask (resp. bid) market order. Selected stocks show the $20\%$ (left) and $80\%$(right) quantiles measured in mean $L^2$ error.}
\label{figure:ProbabilityImbalanceLast}
\end{center}
\end{figure}
It is interesting to observe the strong effect of the sign of the last trade on the imbalance predicting power. While a close to zero imbalance unconditionally indicate a close to $50\%$ probability of an ask trade (see Figure \ref{figure:ProbabilityImbalance}), this dramatically changes when the sign of the last trade comes into play. We now observe on Figure \ref{figure:ProbabilityImbalanceLast} two probability curves, one for each sign of the preceding trade, and these curves highlight a kind of hysteresis effect in sign determination with respect to the imbalance. For example, during a series of bid (resp. ask) market orders, the imbalance level at which an ask (resp. bid) market order becomes more probable is not around $0$, but higher (resp.lower). On the selected graphs the $50\%$ crossing point for the imbalance level is around $\pm0.5$ (this value may be stock and date dependent).

In a further step, we can investigate the role of the spread on these dynamics. In the case of a large spread observation, then priority can be achieved with limit orders, it is thus expected that the imbalance effect will be less pronounced (as found for example in \cite{Stoikov2017}) but that it will interact with $\epsilon$. We can thus use the following ratio model :
\begin{align}
	\lambda^{MA}(t,\vartheta^{MA}) = \lambda_0(t) \exp\left[\vartheta^{MA}_0+\vartheta^{MA}_1 i(t) + \vartheta^{MA}_2 \epsilon(t) + \vartheta^{MA}_3 \epsilon(t)s(t) \right],
	\nonumber \\ 
	\lambda^{MB}(t,\vartheta^{MB}) = \lambda_0(t) \exp\left[\vartheta^{MB}_0+\vartheta^{MB}_1 i(t) + \vartheta^{MB}_2 \epsilon(t) + \vartheta^{MA}_3 \epsilon(t)s(t)\right],
	\label{equation:Model-MBMA-ImbalanceLastSpread}
\end{align}
where $s(t)=+1$ if the observed spread is larger than its mean, and $-1$ if it is smaller. Obviously, one could use the spread value in ticks for finer models, but we choose categorical variable for graphic illustration purposes. Note that the spread distribution is very stable, so that the mean can for example be computed on the previous sample in the case of repeated fits. Note also that the spread distribution is positively skewed and bounded below by zero, so that $s(t)=-1$ can be interpreted as the usual spread case, and $s(t)=+1$ as the large spread case.
Figure \ref{figure:ProbabilityImbalanceLastSpread} plots the updated results for two samples. For representativity, we again select the two samples representing the $20\%$ and $80\%$ quantiles measured in mean $L^2$ error. All 390 fits are available upon request.
\begin{figure}
\begin{center}
\begin{tabular}{cc}
\includegraphics[scale=0.42, page=2]{{Cox-MBMA-ImbalanceLastmarketSpread-MONTHLY-SELECTED}.pdf}
&
\includegraphics[scale=0.42, page=1]{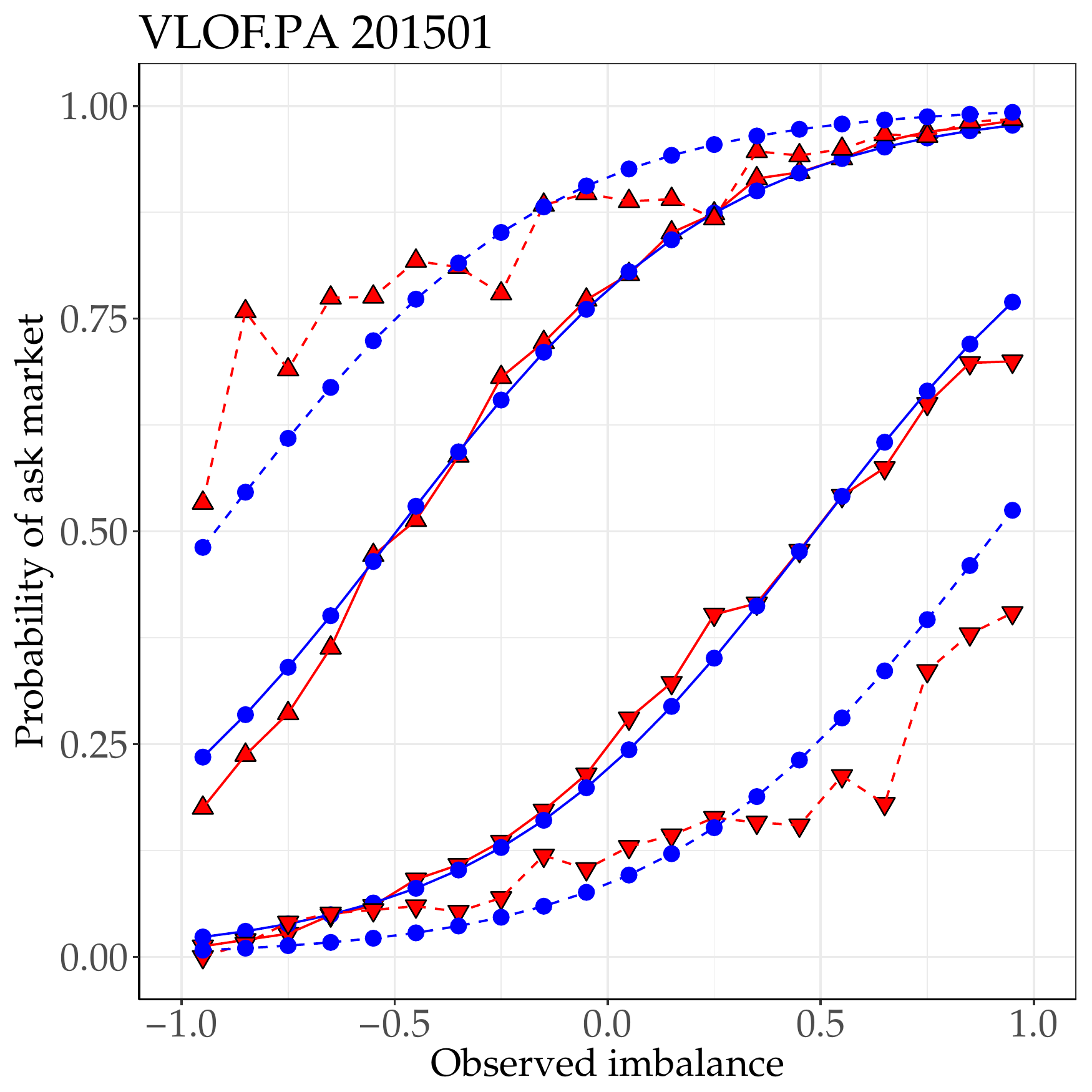}
\end{tabular}
\caption{Empirical probability (in red) and fitted probability (in blue) that the next market order is an ask market order, given the observed imbalance. Upward (resp. downward) triangles indicate that the last trade was an ask (resp. bid) market order. Dashed (resp. full) lines represent large (resp. usual) spread values. Selected stocks show the $20\%$ (left) and $80\%$(right) quantiles measured in mean $L^2$ error.}
\label{figure:ProbabilityImbalanceLastSpread}
\end{center}
\end{figure}
We now have four curves representing the effect of the imbalance on the next trade sign, depending on the last trade sign and the current spread. Empirical curves are noisier as the samples are split into subsamples to compute conditional probabilities. However, we observe as expected that a large spread flattens the imbalance effect and reinforces the last sign importance. On the provided example fits, if the spread is larger than usual and the last trade was a bid, an ask market order becomes more probable than a bid market order when the imbalance nearly reaches $1$, compared to roughly $0.5$ when the spread is at its usual values, and to $0$ when neither the spread nor the last sign are taken into account.

Since we are exploring the role of several covariates, it is natural to apply the information criteria mentioned in Proposition \ref{20180504-1}. Among other possible sequences $a_T$, we use the QAIC for $a_T=2$, the QCAIC for $a_T=\log T+1$, and the QBIC for $a_T=\log T$, all based on the QMLE $\hat{\theta}_\mathbb{K}$. As an illustration, Table \ref{table:QIC} (left panel) shows the number of times various models are selected among the 390 samples, according to these information criteria. 
\begin{table}
\begin{center}
\begin{tabular}{|c|c||ccc||ccc|}
\hline
$\mathsf{d}$ & Model & QAIC & QCAIC & QBIC & QAIC & QCAIC & QBIC \\
\hline
2	& $i$	& 0	& 0	& 0 & 0	& 0	& 0 \\
3	& $i.\epsilon$	& 0	& 2	& 2 & 0	& 2	& 2 \\
4	& $i.\epsilon.s$	& 1	& 1	& 1 & 1	& 1	& 1 \\
4	& $i.\epsilon.\epsilon s$	& 79	& 162	& 158 & 48	& 118	& 114 \\
7	& $i.\epsilon.s+\textrm{int}$	& 310	& 225	& 229 & 218	& 175	& 178 \\
11  & $i.\epsilon.s.\delta+\textrm{int}$ & ---	& ---	& --- & 123	& 94	& 95 \\
\hline
\end{tabular}
\caption{Number of times a model for the intensity of bid and ask market orders is selected among the 390 samples, according to several information criteria. Models are named by the names of the covariates separated with a dot, with the notations defined in the text. "+int" means that all interaction terms are included in the model. E.g., "$i.\epsilon.\epsilon s$" is the model defined at Equations \eqref{equation:Model-MBMA-ImbalanceLastSpread}. Left panel investigates covariates $i,\epsilon$ and $s$ only, while right panel adds the last price $\delta$.}
\label{table:QIC}
\end{center}
\end{table}
Several comments can be made about this illustration. First of all, it turns out that the simplest model depending only on the imbalance is never selected. Recall for example that the weighted mid-price, commonly used in microstructure as a proxy for a "future" price, depends only on the imbalance. Our observation pleads for the incorporation of more covariates in defining such proxies. It also highlights that, as explained above, the spread acts primarily on the intensity of submission in interaction with the last sign: model with covariates $i, \epsilon$ and $s$ is (nearly) never selected, in contrast to the model with covariates $i, \epsilon$ and $\epsilon s$. Finally, we observe as expected that QAIC has a tendency to select the model with the greatest number of parameters~; however, using QCAIC or QBIC, the simple model of Equation \eqref{equation:Model-MBMA-ImbalanceLastSpread} is still selected on more than $40\%$ of the samples over the full model with all terms. These results show that in future works further investigations could be made in determining relevant covariates and possibly analyze a stock- or time- dependency. Among several possibilities, Table \ref{table:QIC} (right panel) gives selection results when we add the last price movement (denoted $\delta$) to ratio model.

Our investigation on the influence of the imbalance signal on the intensities of bid and ask market orders can also illustrate the use of the penalized QMLE described in Section \ref{section:Penalization}. In this example, we use the ratio model to try to decide whether traders use the imbalance as a trading signal on a given stock looking at the first level only, or at the first two levels, etc. Let us denote $i_k(t)$ the imbalance observed a time $t$ computed using the \emph{cumulative} quantities at the best quotes up to the $k$-th level, $k=1,\ldots,10$. $i_1$ is thus the imbalance previously investigated. For $k=2,\ldots,10$, let $\Delta i_k = i_k - i_{k-1}$ the corrective term between the imbalances computed with $k$ and $k-1$ limits.
We extend the basic model \eqref{equation:Model-MBMA-Imbalance} to include the standard imbalance as well as all the corrective terms as potential trading signals:
\begin{equation}
	\lambda^{T}(t,\theta^{T}) = \lambda_0(t) \exp\left[\theta^{T}_0 + \theta^{T}_1 i_1(t) + \sum_{k=2}^{10}\theta^{T}_k \Delta i_k(t)\right],
	\label{equation:Model-MBMA-AllImbalances}
\end{equation}
where $T\in\{MB,MA\}$. Let $\theta_k=\theta^{MA}_k-\theta^{MB}_k$, $k=1,10$ be our ratio parameters, to be estimated by likelihood maximization.
We estimate the model using both standard quasi-likelihood maximization, as well as a penalized estimation using the penarized QMLE $\hat{\theta}_\lambda$ of Equations \eqref{20180504-10} and \eqref{20180504-11}.
We compute daily fits of the model and plot the mean values of the estimated parameters for each stock.
Figure \ref{figure:AllImbalances-Delta-Penalization} shows the fitted parameters $\theta_k$ as a function of the level $k$, as well as the $95\%$ Gaussian empirical confidence interval computed with the empirical standard deviation. For brevity only two stocks are shown, and for simplicty we choose the same ones as in Figure \ref{figure:ProbabilityImbalanceLastSpread} but results for all 36 stocks are similar. 
\begin{figure}
\begin{center}
\begin{tabular}{cc}
\includegraphics[scale=0.42, page=1]{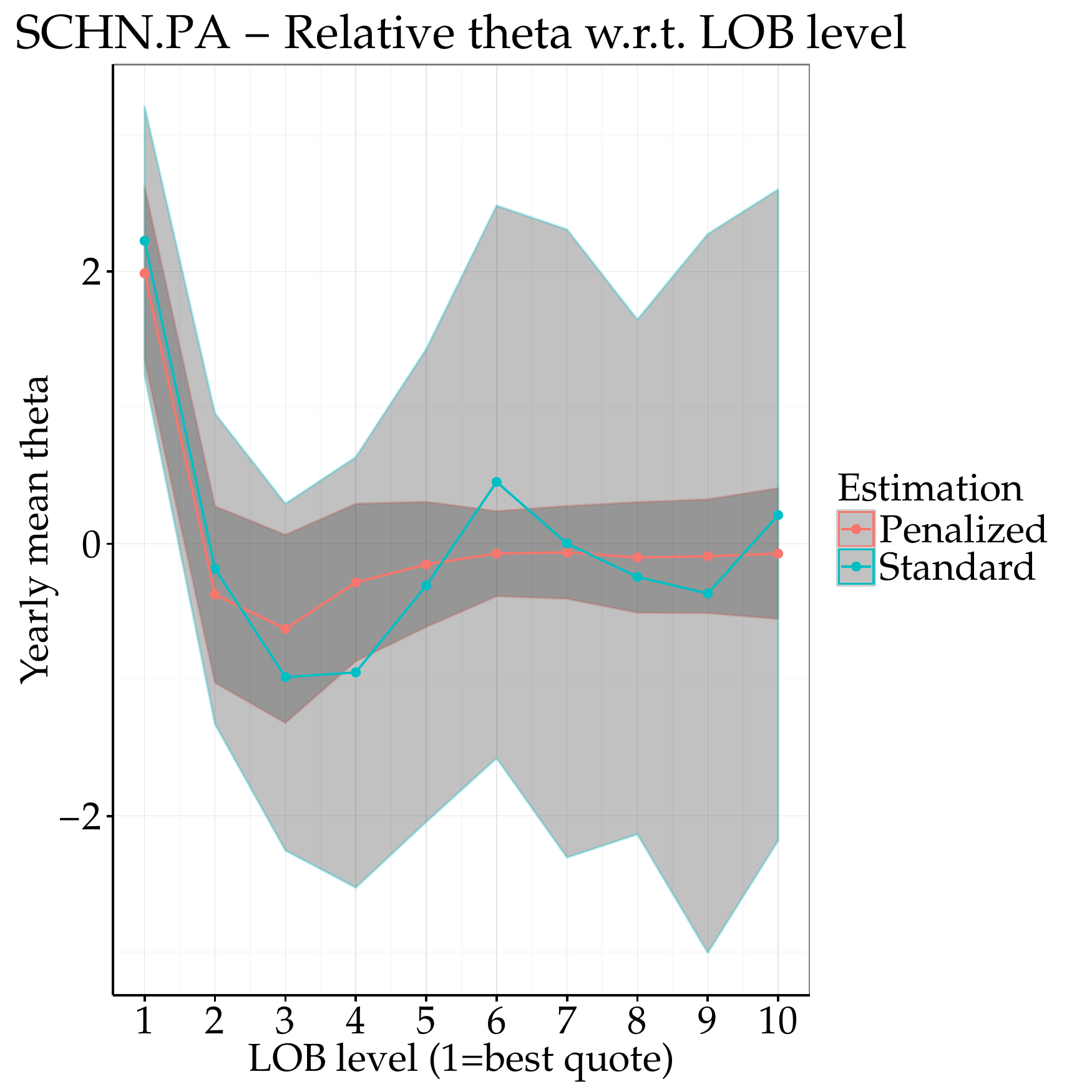}
&
\includegraphics[scale=0.42, page=2]{{Cox-MBMA-allImbalances-Delta-Penalization-BuildGraphs}.pdf}
\end{tabular}
\caption{Mean daily estimate (curves) and $95\%$ confidence interval (shaded areas) of $\theta_k$ for the imbalance model with corrective terms of Equation \eqref{equation:Model-MBMA-AllImbalances}. Same stocks as in Figure \ref{figure:ProbabilityImbalanceLastSpread}. In these numerical examples, we use $q=1$ and $\lambda=50 T^{-\frac{1}{2}}$.}
\label{figure:AllImbalances-Delta-Penalization}
\end{center}
\end{figure}
It turns out that the parameters associated to the imbalance at level 1 is always very significant, but that the additional information given by the corrective terms $\Delta i_k, k=2,\ldots,10$ is not translated into significant $\theta_k$, $k=2,\ldots,10$. Penalization greatly helps reducing the confidence intervals of estimates.
This results thus seem to be in favor of arguing that traders use imbalance at the first level as a significant trading signal, but do not significantly use information of higher levels.

\subsection{Spread and order flows at the best quotes}
\label{subsec:MLCSpread}

Spread as a potential trading signal seems less investigated in the microstructure literature than other factors. One potential reason is that the spread of many highly traded stocks is almost always equal to one tick (so-called large tick stocks). This property has given birth to specific models of such limit order books assuming a constant spread.
However, in the studied sample of 36 stocks traded on the Paris stock exchange in 2015, spread cannot be assume equal to one. As such, it must be considered a possible trading signal, influencing the order flows.
If the observed spread is large, then a trader wanting to buy a share would rather submit a bid limit order inside the spread than an ask market order. That way the trader will secure priority for the next sell market order, while obtaining a lower price. A market order should only be used by an impatient trader, or a trader believing in an incoming and lasting upward market movement. On the contrary, a spread equal to one tick prevents the use of limit orders to secure priority in future executions. Such mechanisms have been observed in \cite{MuniTokeYoshida2017} where data analysis shows that the empirical intensity of market orders increases when the spread decreases to one tick.

We can therefore propose a version of our ratio model to investigate the influence of the spread on orders flows occurring at the best quotes. 
Let us consider a model of best quotes of the limit order book where market orders, limit orders and cancellations are submitted with intensities $\lambda^M$, $\lambda^L$ and $\lambda^C$ respectively, all sharing an unobserved baseline intensity $\lambda_0(t)$ and depending on the observed spread $s(t)\in\mathbb{N^*}$ (i.e expressed in number of ticks) :
\begin{align}
	\lambda^{M}(t,\vartheta^{M}) & = \lambda_0(t) \exp\left[\vartheta^{M}_0+\vartheta^{M}_1 \log s(t) +\vartheta^{M}_2 (\log s(t))^2\right],
	\\ 
	\lambda^{L}(t,\vartheta^{L}) & = \lambda_0(t) \exp\left[\vartheta^{L}_0+\vartheta^{L}_1 \log s(t) +\vartheta^{L}_2 (\log s(t))^2\right],
	\\ 
	\lambda^{C}(t,\vartheta^{C}) & = \lambda_0(t) \exp\left[\vartheta^{C}_0+\vartheta^{C}_1 \log s(t) +\vartheta^{C}_2 (\log s(t))^2\right].
	\label{equation:Model-MLC-Spread}
\end{align}
Letting $\theta^L_j = \vartheta^{L}_j - \vartheta^{M}_j$ and $\theta^C_j = \vartheta^{C}_j - \vartheta^{M}_j$, $j=0,1,2$, our intensities ratio model is easily estimated by maximizing the quasi-log likelihood.

We estimate the model for each stock and each trading days, which gives after data cleaning 8052 different samples and associated model fits.
As in the previous cases, we compute the intensities ratios to estimate the probabilities of each event given the observed spread. Figure \ref{figure:ProbabilitySpread} plots the probabilities of market orders, limit orders and cancellations occurring at the best quotes given the observed spread, and compare them to the empirical probabilities computed on the sample. Again for the sake of brevity and representativity we rank the fits by increasing mean $L^2$ error between model and data, and select four stocks representing the $20\%$, $40\%$, $60\%$ and $80\%$ quantiles of the error distribution.
\begin{figure}
\begin{center}
\begin{tabular}{cc}
\includegraphics[scale=0.42, page=4]{{Cox-MLC-Spread-LOG-DAILY-SELECTED}.pdf}
&
\includegraphics[scale=0.42, page=3]{{Cox-MLC-Spread-LOG-DAILY-SELECTED}.pdf}
\\
\includegraphics[scale=0.42, page=2]{{Cox-MLC-Spread-LOG-DAILY-SELECTED}.pdf}
&
\includegraphics[scale=0.42, page=1]{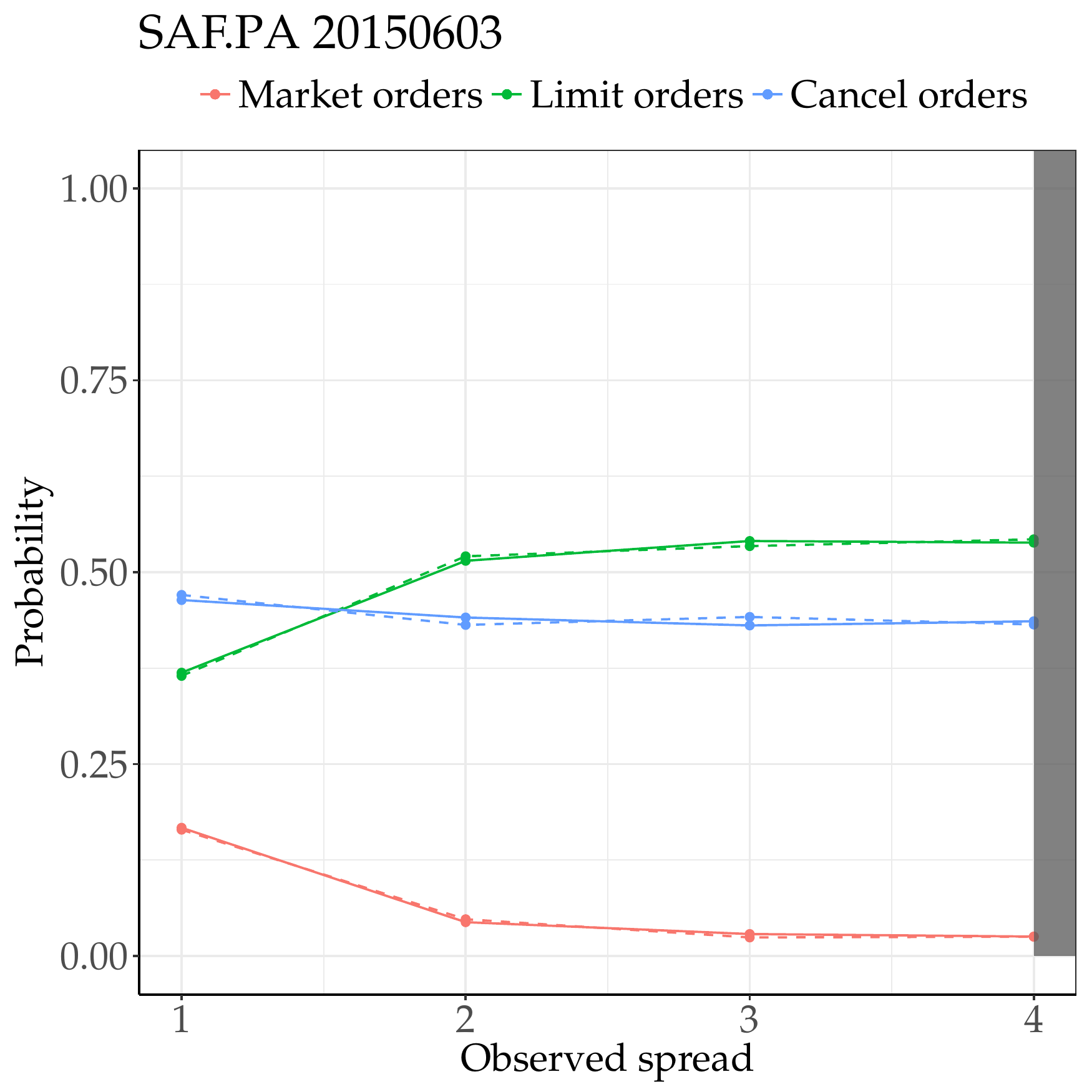}
\end{tabular}
\caption{Empirical probability (dashed lines) and fitted probability (full lines) that the next order is a market order (red), a limit order (blue) or an cancellation (green), given the observed spread. $x$-axis spans more than $99\%$ of the empirical spread distribution, and the shaded area on the right signal the $95\%$ and $99\%$ quantiles of this distribution. Selected samples represent the $20\%$ (top left), $40\%$ (top right), $60\%$ (bottom left) and $80\%$ (bottom right) quantiles of the mean $L^2$ error distribution among the 8052 tested samples.}
\label{figure:ProbabilitySpread}
\end{center}
\end{figure}
It is satisfying to observe that the model provide good fits in a wide range of spread distribution, from a large tick stock (Figure \ref{figure:ProbabilitySpread}, bottom left) where the spread is almost always equal to one tick, to the more interesting case of small tick stocks (e.g., Figure \ref{figure:ProbabilitySpread}, top right), where the probability is correctly estimated up to 9 ticks and more.

\subsection{Equilibrium behaviour with respect to the queues sizes}
\label{subsec:LalphaCalphaqalpha}

We propose a final illustration of the ratio model to investigate the role of the observed queue size in determining the flows of limit orders and cancellations. Some empirical observations of the role of the queue size can be observed in \cite{Huang2015} as well as in \cite{MuniTokeYoshida2017}. A basic equilibrium argument suggests that when a queue size at a given price is small, then providing liquidity with limit orders is interesting as it secures a relative priority for the trader. On the contrary, when a queue size is large, then cancellations should be more frequent than limit orders. Furthermore, such an argument should be "more" valid inside the book, where only limit and market orders are allowed, than closer to the best quotes, where market orders remove important portions of liquidity and other trading signals, such as the spread and the imbalance previously investigated, play a significant role.

We investigate these insights by building one simple ratio model for limit orders and cancellations occurring inside the book, from level 2 to 10. Level 1 (best quote) is left aside as its dynamics also involves market orders. For each level $\alpha\in\{2,\ldots,10\}$, let $N^{L_\alpha}$ and $N^{C_\alpha}$ be the counting processes of limits orders and cancellations occurring at this level, with intensities:
\begin{align}
	\lambda^{L_\alpha}(t,\theta^{L_\alpha}) = \lambda_0(t) \exp\left[\theta^{L_\alpha}_0+\theta^{L_\alpha}_1 \log q_\alpha(t) +\theta^{L_\alpha}_2 (\log q_\alpha(t))^2 \right],
	\nonumber	\\ 
	\lambda^{C_\alpha}(t,\theta^{C_\alpha}) = \lambda_0(t) \exp\left[\theta^{C_\alpha}_0+\theta^{C_\alpha}_1 \log q_\alpha(t) +\theta^{C_\alpha}_2 (\log q_\alpha(t))^2 \right].
	\label{equation:Model-LalphaCalpha-qalpha}
\end{align}
where $\lambda_0(t)$ is a potentially random baseline intensity and $q_\alpha(t)$ is the volume standing in the book at the level of submission in group $\alpha$ and time $t$. $q_\alpha(t)$ is expressed in integer multiple of the median trade size (see e.g., \cite{Huang2015,MuniTokeYoshida2017} about this normalization).
Let $\theta^\alpha_j=\theta^{L_\alpha}_j-\theta^{C_\alpha}_j$, $j=0,2$, $\alpha\in\{2,\ldots,10\}$, be the relative parameters of the ratio model, to be estimated by the maximization of quasi-log likelihood.

As in the previous examples, we estimate the model for each stock and each month, and then compute the intensities ratios to estimate the probabilities of each event given the observed volume of the level of submission. Each monthly fit gives 9 fits, one for each level, but for the sake of readability, we plot only three levels: level 2 at the head of the book, level 5 at mid-depth and level 8 at the back of the book. Figure \ref{figure:ProbabilityQalpha} plots the probabilities of limit orders occurring inside the book at these three levels, given the observed volume of the level, and compare them to the empirical probabilities computed on the sample. We show four examples of fits (again selected by computing the quantiles of the mean $L^2$ error of the fits), but all monthly fits are quite similar and available.
\begin{figure}
\begin{center}
\begin{tabular}{cc}
\includegraphics[scale=0.42, page=4]{{Cox-LalphaCalpha-qalpha-MONTHLY-SELECTED}.pdf}
&
\includegraphics[scale=0.42, page=3]{{Cox-LalphaCalpha-qalpha-MONTHLY-SELECTED}.pdf}
\\
\includegraphics[scale=0.42, page=2]{{Cox-LalphaCalpha-qalpha-MONTHLY-SELECTED}.pdf}
&
\includegraphics[scale=0.42, page=1]{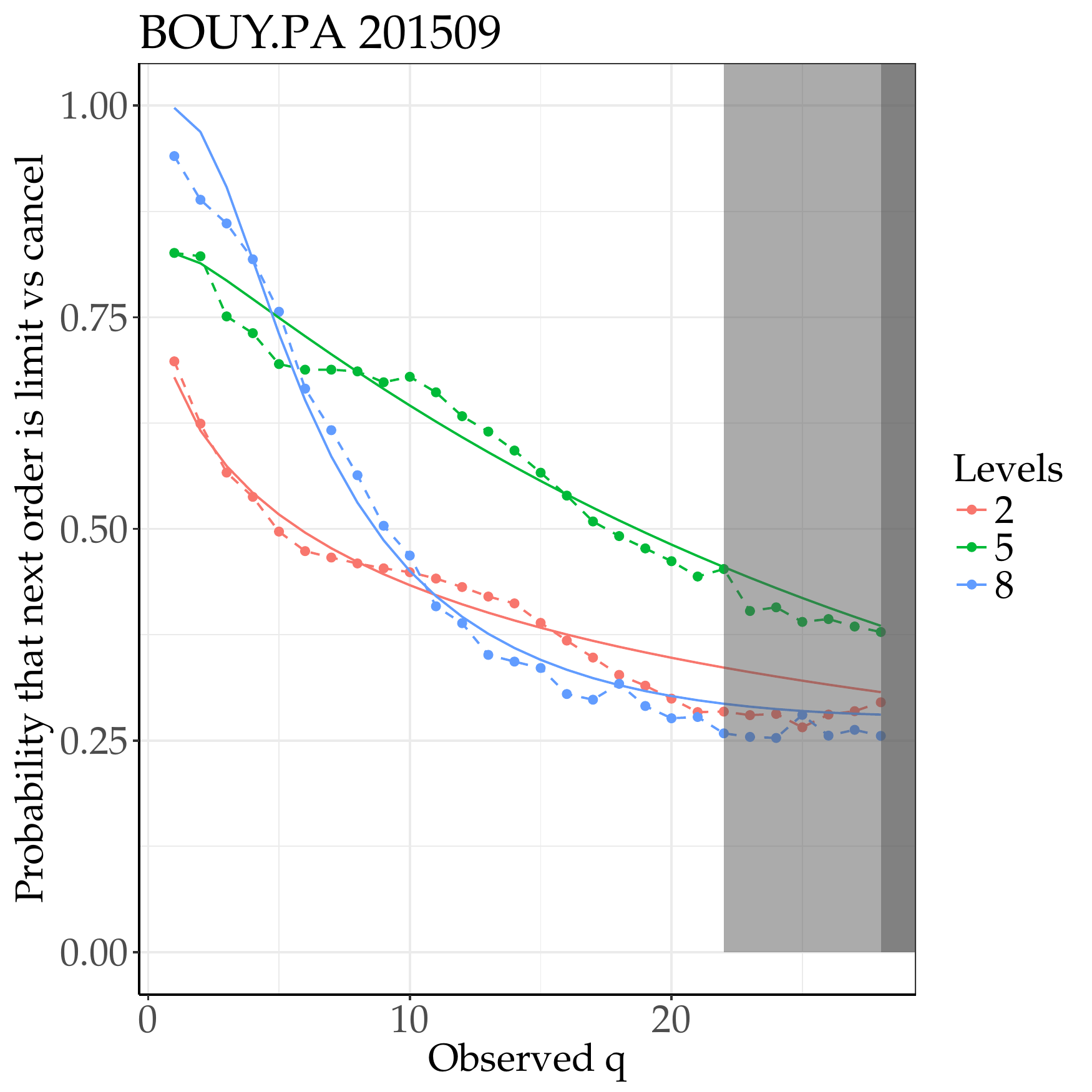}
\end{tabular}
\caption{Empirical probability (dots and dashed lines) and fitted probability (lines) that the next order is a limit order, given the observed level of the queue of submission. $x$-axis spans more than $99\%$ of the empirical volume distribution (at all levels), and the shaded area on the right signal the $95\%$ and $99\%$ quantiles of this distribution.}
\label{figure:ProbabilityQalpha}
\end{center}
\end{figure}
It is interesting that the probabilities of occurrence of a limit order as a function of the observed volume show similar patterns across stocks. The model obviously recovers the expected equilibrium property that decreases the interest (hence the probability) of a limit order when liquidity is already here (i.e. $q_{\alpha}$ is large). Furthermore, the use of a common $x$-axis makes the the probability curves reflect the general shape of a limit order book. Using a rule of thumb that the equilibrium size of each level should be roughly at the crossing of the $.5$ probability line, the respective position of the three probability curves reflects the well-known humped shaped of the limit order book (\cite{Bouchaud2002}): the mid-book (level 5) is fatter than the head (level 2) and tail (level 8) of the book.

\section{Empirical results : Prediction}
\label{section:EmpiricalResultsPrediction}

Empirical results from Section \ref{section:EmpiricalResults} are in-sample analysis aimed at providing new descriptions of some dependencies observed on financial markets. We now turn to the possible use of the ratio model as a prediction tool.
We consider the problem of the sign of the next trade, previously introduced in Section \ref{subsec:MBMAImbalance}. The ratio model has helped us model how this sign depends on the state of the order book, through e.g., the imbalance, the spread or the last trade sign.
Time series of trade signs are known to be highly auto-correlated, and to be well described by Hawkes processes. In this setting, market orders can be described by a two-dimensional Hawkes process $Z=(Z^B, Z^A)$ with exponential kernels, its conditional intensity $\lambda_Z=(\lambda_{Z^B}, \lambda_{Z^A})$ being written in vector notations as
\begin{equation}
	\lambda_{Z}(t) = \mu_Z + \int_0^t K(t-s)\,dZ_s,
\label{eq:Prediction-HawkesIntensity}
\end{equation}
where $\mu_Z=(\mu_{Z^B}, \mu_{Z^A})$ is the baseline intensity and the kernel matrix 
\begin{equation}
	K(t) = \left(\begin{matrix}
		\alpha_{BB} e^{-\beta_{BB} t} & \alpha_{BA} e^{-\beta_{BA} t}
		\\
		\alpha_{AB} e^{-\beta_{AB} t} & \alpha_{AA} e^{-\beta_{AA} t}
		\end{matrix}\right)
\label{eq:Prediction-HawkesKernel}
\end{equation}
describes the self- and cross- excitation parts.

A simple ratio model with only current observations of the state and not taking into account the history of the order flow (except for the last trade sign) can probably not compete with the Hawkes description.
However, we can incorporate some history into covariates of the ratio model. One may for example consider the following covariates:
\begin{align}
	H_A(t) & = \log\left(\mu_A + \int_0^t \alpha_A e^{-\beta_A(t-s)}\,dN^A_s\right),
	\\
	H_B(t) & = \log\left(\mu_B + \int_0^t \alpha_B e^{-\beta_B(t-s)}\,dN^B_s\right),
\end{align}
where $(N^B, N^A)$ counts the number of bid and ask market orders. These covariates add some self-exciting history into the ratio model.

We can thus examine different methods to predict the trade sign, given that one trade is observed a given time. This is of course a theoretical exercise, as we predict the trade sign with all the information available just before its occurrence, not taking into account latency, information delays, reaction times, etc., that would hinder the performances in a more realistic setting. We test seven methods for this exercise:
\begin{itemize}
	\item Last : the trade sign is set to be the same as the last observed trade sign ;
	\item Imbalance : the trade sign is set to $-1$ if the imbalance observed before its submission is negative, $+1$ otherwise;
	\item Hawkes Full : at each time we compute the intensity $\lambda_Z$ of the Hawkes process described at Equations \eqref{eq:Prediction-HawkesIntensity}-\eqref{eq:Prediction-HawkesKernel}, given the observed history, and set the sign to be $-1$ if $\lambda_{Z^B}>\lambda_{Z^A}$, $+1$ otherwise.
	\item Hawkes NoCross : same as the previous method, except that we independently fit two Hawkes processes for the bid and ask market orders, i.e. we set $\alpha_{BA}=\alpha_{AB}=0$;
	\item Ratio $(i,\epsilon,\epsilon s)$ : we use the probabilities given by the basic ratio model of Equation \eqref{equation:Model-MBMA-ImbalanceLastSpread} and set the sign to be $-1$ if the probability of an ask market order given by the ratio model is lower than $0.5$, and $+1$ otherwise ;
	\item Ratio $(H_B,H_A)$ : same as the previous method, but we use only the covariates $(H_B, H_A)$ in the ratio model, instead of the covariates describing the state of the book ;
	\item Ratio $(H_B,H_A, i,\epsilon,\epsilon s)$ : same as the previous method, but using both $(H_A, H_B)$ and the covariates $(i,\epsilon,\epsilon s)$ in the ratio model.
\end{itemize}
We test these methods on the sample described at Section \ref{subsec:Data}. ``Last'' and ``Imbalance'' methods do no require any calibration. For the other methods, calibration is carried out on the trading day preceding the day at which prediction performance is evaluated. Models ``Hawkes Full'' and ``Hawkes NoCross'' are fitted with a maximum-likelihood estimation (see e.g. \cite{Ogata1978}, \cite{Ozaki1979} and later \cite{Bowsher2007}, \cite{Bacry2013} or \cite{Pomponio2012} in a financial setting). The parameters $(\mu_A, \alpha_A, \beta_A)$ and $(\mu_B, \alpha_B, \beta_B)$ used to compute the covariates $H_A$ and $H_B$ are obtained by the same method, i.e. the estimated values of the parameters $(\mu_A, \alpha_A, \beta_A)$ used to compute $H_A$ correspond to the estimated values $(\mu_{Z^A}, \alpha_{AA}, \beta_{AA})$ of the ``Hawkes NoCross'' model.
Note that the preceding day in the sample means Friday for Mondays, and possibly some other previous day in case of missing trading days in the data. Recall that this sample represents roughly 54 millions trades.

Results are illustrated in Figure \ref{figure:Prediction-General}.
\begin{figure}
\begin{center}
\includegraphics[width=0.7\textwidth]{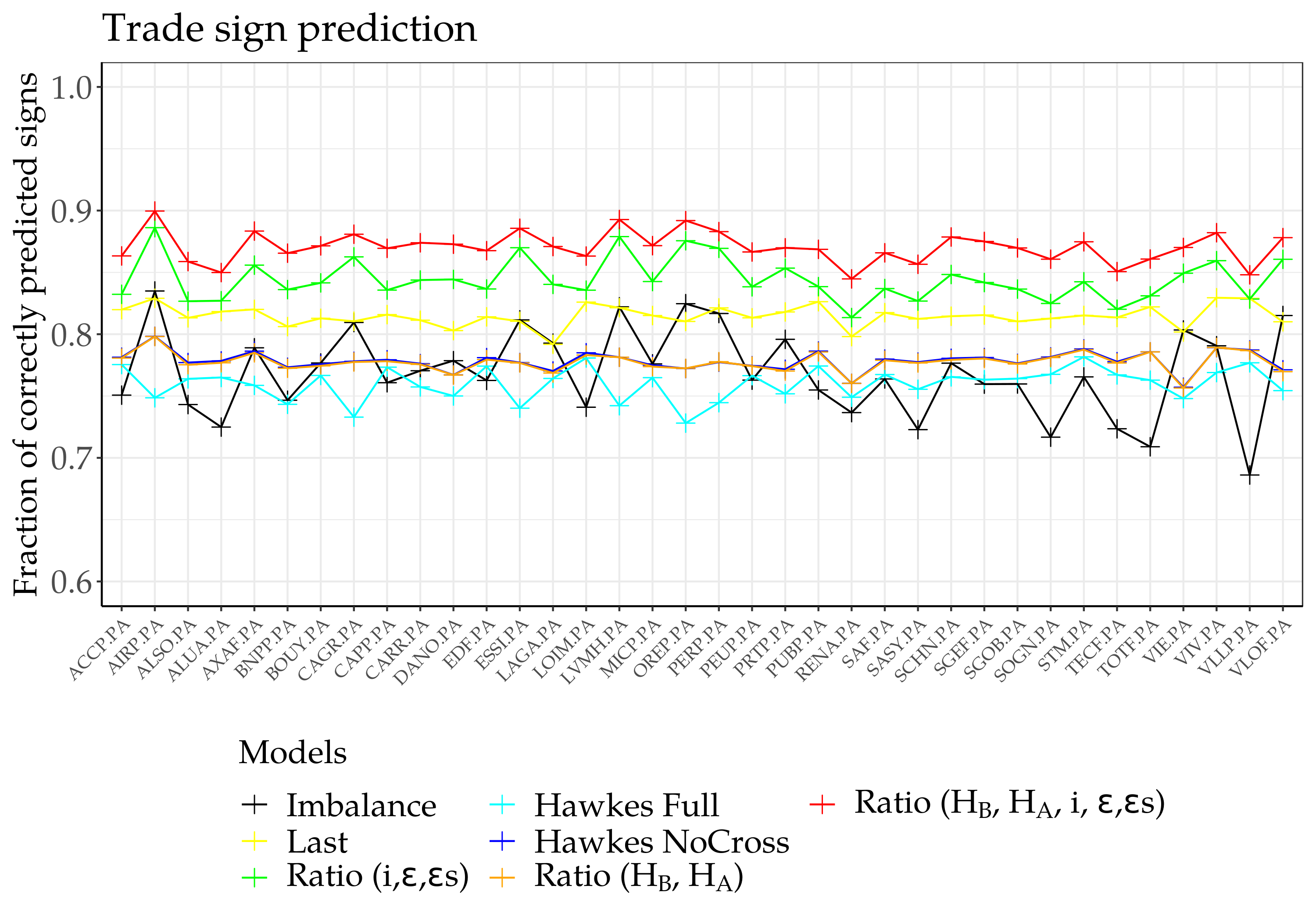}
\caption{Fraction of correctly predicted trade sign per method and per stock, averaged on all available trading days.}
\label{figure:Prediction-General}
\end{center}
\end{figure}
For all stocks, the ``Last'' indicator correctly predicts more than $80\%$ of the trade signs, e.g. less than one trade out of 5 has a sign different from the preceding trade. The ``Imbalance'' indicator is less performant, signing correctly about $70\%-80\%$ of the trades. The model ``Ratio $(i,\epsilon,\epsilon s)$'' of Section \ref{subsec:MBMAImbalance} builds on these two indicators to improve the signing performance to about $85\%$. Hawkes models perform as the ``Imbalance'': the performance of the ``Hawkes Full'' method lies in the range $73\%-78\%$, less than the ``Hawkes NoCross'' model that actually does better in this prediction exercise, always around $75\%-80\%$. This result is an illustration of the observation that parcimony is often crucial for prediction.
Now, it is interesting to observe that the ratio model is actually able to match and improve the Hawkes performances. Using only the covariates $(H_B,H_A)$, the ``Ratio $(H_B,H_A)$'' method actually mimicks the Hawkes process, the two performances curves being very close. The advantage of the ratio model is the possibility to consider both the history covariates $(H_B,H_A)$ and the state covariates $(i,\epsilon,\epsilon s)$. The model ``Ratio $(H_B,H_A,i,\epsilon,\epsilon s)$'' turns out to deliver the best prediction of the trade sign, improving the performance of the ``Ratio $(i,\epsilon,\epsilon s)$'' model by a bit more than $2.5\%$ in average on all the stock-days.

We can analyse a bit further these performances. We now focus on the prediction of a sign change, i.e. we compute the performance using only the trades that have a sign different from the previous trade (roughly $20\%$ of the sample, given the observation above). Results are averaged for each stock on Figure \ref{figure:Prediction-SignChange}.
\begin{figure}
\begin{center}
\includegraphics[width=0.7\textwidth]{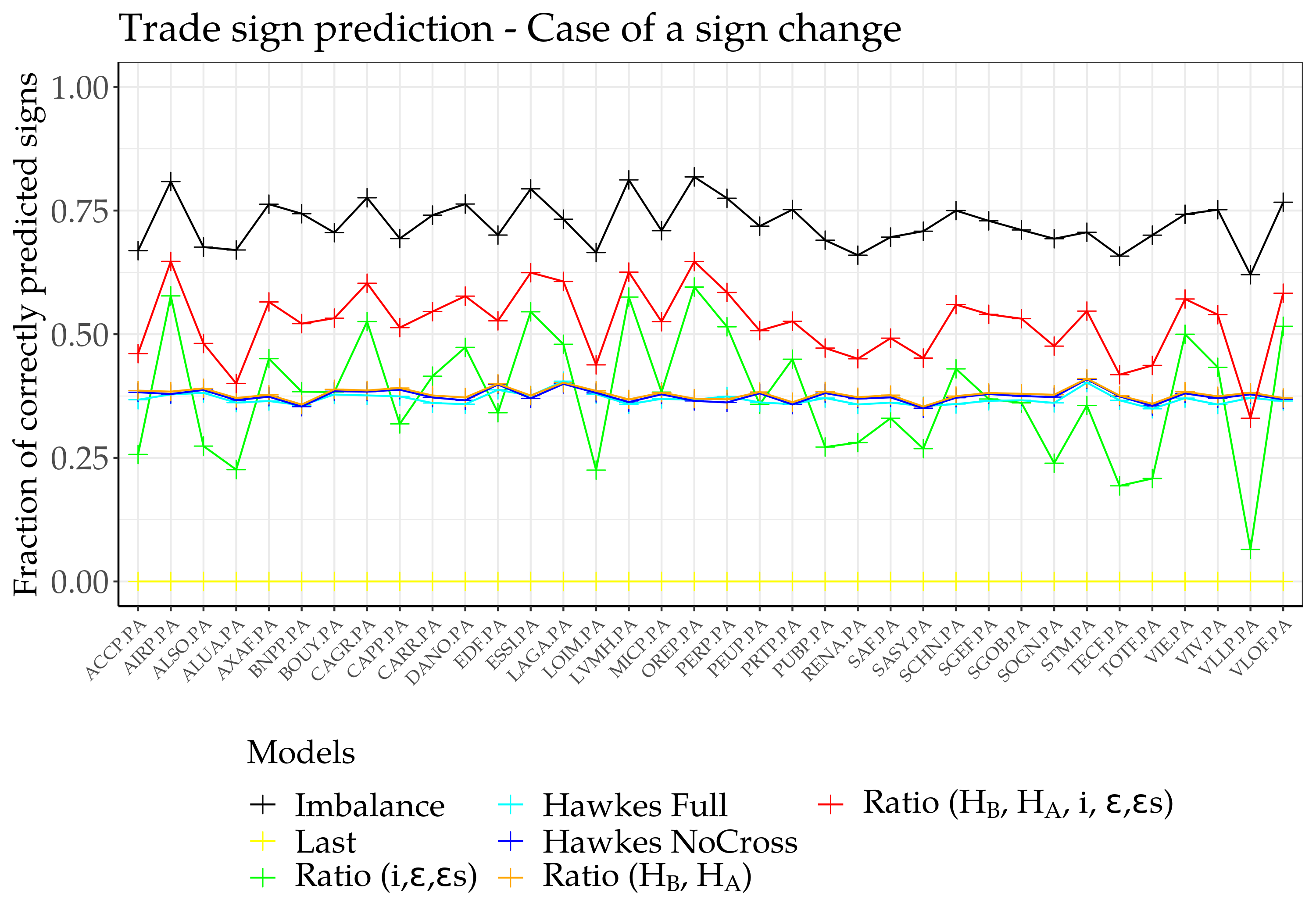}
\caption{Fraction of correctly predicted trade sign per method and per stock, averaged on all available trading days.}
\label{figure:Prediction-SignChange}
\end{center}
\end{figure}
The performance of the ``Imbalance'' signal is roughly similar. The ``Last'' indicator has obviously a performance equal to $0$, as it never predicts a sign change. The ratio model without any history also performs poorly (around $40\%$) in this specific subsample, since its only history relies on the last trade sign. Both Hawkes model are also less performant, always even worse than the basic ``Imbalance'' signal. The interesting point is that in this specific subsample, the combination of the Hawkes covariates and the state covariates $(i,\epsilon,\epsilon s)$ described in Section \ref{subsec:MBMAImbalance} significantly improves the Hawkes models and the ``Ratio $(i,\epsilon,\epsilon s)$'' model. The average performance increase on all the stock-days of the sub-sample for the ``Ratio $(H_B,H_A,i,\epsilon,\epsilon s)$'' model with respect to the best Hawkes model is more than $13\%$. In brief, trade signs are quite well represented by Hawkes processes, but the addition of a state-dependency by means of a ratio model allows a better tracking of the sign changes, explaining an overall better performance.

\section{Conclusion}

We have presented a model based on ratios of Cox-type intensities sharing a common, possibly random, baseline intensity. Consistency and asymptotic normality of the estimators have been proved. Such a model may be very useful in cases where one tries to investigate the role of given covariates on point processes in a fluctuating environment, but in which (some of) these fluctuations are assumed to equally influence all the processes under scrutiny. This is therefore a plausible framework in finance, where global market activity varies wildly during the trading day. The proposed setting removes this common baseline intensity from the estimation procedure, allowing to focus on the role of covariates. Using this method we have in particular been able to highlight how signals such as market imbalance, bid-ask spread and sizes of limit order book queues do influence trading activity. This framework may hopefully be helpful in many other studies in high-frequency finance. Several directions of research and applications may indeed be followed, among which model selection and the identification of significant trading signals.

\section*{Acknowledgements}
This work was in part supported by Japan Science and Technology Agency CREST JPMJCR14D7; Japan Society for the Promotion of Science Grants-in-Aid for Scientific Research No. 17H01702 (Scientific Research), No. 26540011 (Challenging Exploratory Research); and by a Cooperative Research Program of the Institute of Statistical Mathematics. 

\appendix

\section{Mathematical proofs}

\subsection{Proof of Theorem \ref{170925-2}}
We have 
\begin{equation} 
r^i(t,\theta) 
= 
\rho^i(\mathbb X(t),\theta)\qquad(i\in\mathbb I,\>\theta\in\Theta)
\end{equation}
and 
\begin{equation} 
\lambda^i(t,\vartheta^*)
=
\rho^i(\mathbb X(t),\theta^*)\Lambda(\lambda_0(t),\mathbb X(t)).
\qquad(i\in\mathbb I)
\end{equation}

Simple calculus yields 
\begin{equation}
\partial_{\theta^\alpha}\log \rho^i(x,\theta)
=
\big\{\delta_{\alpha,i}-\rho^\alpha(x,\theta)\big\}x \qquad(i\in\mathbb I,\>\alpha\in\mathbb I_0,\>\theta\in\mathbb R^\mathsf p)
\end{equation}
and 
\begin{equation} 
\partial_{\theta^{\alpha'}}\partial_{\theta^\alpha}\log \rho^i(x,\theta)
=
-\big\{\delta_{\alpha',\alpha}\rho^\alpha(x,\theta)-\rho^{\alpha'}(x,\theta)\rho^\alpha(x,\theta)\big\}x^{\otimes2}\qquad(i\in\mathbb I;\>\alpha',\alpha\in\mathbb I_0,\>\theta\in\mathbb R^\mathsf p)
\end{equation}
where $\theta^\alpha=(\theta^\alpha_j)_{j\in\mathbb J}$. 
In other words, 
\begin{equation} 
\partial_{\theta^{\alpha'}}\partial_{\theta^\alpha}\log \rho^i(x,\theta)
=
-{\sf V}(x,\theta)_{\alpha',\alpha}x^{\otimes2}\qquad(i\in\mathbb I;\>\alpha',\alpha\in\mathbb I_0,\>\theta\in\mathbb R^\mathsf p).
\end{equation}

For the closed convex hull $\overline{\mathcal C}[\Theta]$ of $\Theta$, let 
\begin{equation} 
\overline{\mathbb Y}(\theta) 
= 
E\bigg[\sum_{i\in\mathbb I}\log\frac{\rho^i(\mathbb X(0),\theta)}{\rho^i(\mathbb X(0),\theta^*)}\ {\rho^i(\mathbb X(0),\theta^*)}\Lambda(\lambda_0(0),\mathbb X(0))
\bigg]\qquad(\theta\in\overline{\mathcal C}[\Theta])
\end{equation}
Then 
\begin{equation} 
\overline{\mathbb Y}(\theta) 
\leq
0\qquad(\theta\in\overline{\mathcal C}[\Theta])
\end{equation}
and the equality holds if and only if 
\begin{equation}\label{170824-1} 
\rho^i(\mathbb X(0),\theta) = \rho^i(\mathbb X(0),\theta^*)\qquad \lambda_0(0)dP\text{-}a.e.\quad(\forall i\in\mathbb I).
\end{equation}
Condition (\ref{170824-1}) implies that 
\begin{equation} 
\exp\big(\mathbb X(0)[\theta^i]-\mathbb X(0)[\theta^{*i}]\big)
= 
\exp\big(\mathbb X(0)[\theta^0]-\mathbb X(0)[\theta^{*0}]\big)
\\= 1 
\qquad \lambda_0(0)dP\text{-}a.e.\quad(\forall i\in\mathbb I)
\end{equation}
due to $\theta^0=\theta^{*0}=0$. 
Therefore, 
\begin{equation} 
E\bigg[{\sf V}_0(\mathbb X(0))_{i,i}\big(\mathbb X(0)\big[\theta^i-\theta^{*i}\big]\big)^2
\Lambda(\lambda_0(0),\mathbb X(0))%1_{\{\lambda_0(0)>0\}}
\bigg]
=0\quad(\forall i\in\mathbb I_0)
\end{equation}
and hence $\theta^i=\theta^{*i}$ for all $i\in\mathbb I_0$ due to Condition [A3] 
applied to $u=\big(\delta_{i,i'}(\theta^i_j-\theta^{*i}_j)\big)_{i'\in\mathbb I_0,j\in\mathbb J}$. 

We see $\Gamma=-\partial_\theta^2\overline{\mathbb Y}(\theta^*)$.  
Since $\Gamma$ is positive definite and $\overline{\mathbb Y}(\theta)\not=0$ 
for all $\theta\in\overline{\Theta}\setminus\{\theta^*\}$, 
there exists a constant $\chi_0>0$ such that 
\begin{equation}\label{290925-7} 
\overline{\mathbb Y}(\theta) = \overline{\mathbb Y}(\theta)-\overline{\mathbb Y}(\theta^*) \leq -\chi_0\big|\theta-\theta^*\big|^2
\end{equation}
for all $\theta\in\overline{\mathcal C}[\Theta]$. 

Let 
$\mathbb U_T=\{u\in\mathbb R^\mathsf p;\>\theta^\dagger_T(u)\in\Theta\}$, where $\theta^\dagger_T(u)=\theta^*+T^{-1/2}u$.  
The quasi-likelihood ratio random field is define by 
\begin{equation} 
\mathbb Z_T(u) 
=
\exp\big(\mathbb H_T(\theta^*+T^{-1/2}u)-\mathbb H_T(\theta^*)\big)\quad(u\in\mathbb U_T).
\end{equation}
We will use also 
\begin{equation} 
\overline{\mathbb Z}_T(u) = \exp\big(\mathbb H_T(\theta^*+T^{-1/2}u)-\mathbb H_T(\theta^*)\big)\quad(u\in\mathcal C[\mathbb U_T]).
\end{equation}
For this extension, we notice $r^i(t,\theta)$ and hence $\mathbb H_T$ is 
naturally extended to $\mathcal C[\mathbb U_T]$. 
The random field $\overline{\mathbb Z}_T$ will be used to obtain the so-called 
polynomial type large deviation inequality for $\mathbb Z_T$. 

Define $\overline{\mathbb Y}_T$ by 
\begin{equation}
\overline{\mathbb Y}_T(\theta) = T^{-1}\big(\mathbb H_T(\theta)-\mathbb H_T(\theta^*)\big)
\qquad(\theta\in\mathcal C[\Theta])
\end{equation}
with the naturally extended $\mathbb H_T$. 
Define a $\mathsf p$-dimensional random variable $\Delta_T$ and a $\mathsf p\times\mathsf p$ random matrix $\Gamma_T$ by 
\begin{equation} 
\Delta_T = T^{-1/2}\partial_\theta\mathbb H_T(\theta^*)
\end{equation}
and 
\begin{equation} 
\Gamma_T= -T^{-1}\partial_\theta^2\mathbb H_T(\theta^*)
\end{equation}
respectively. 
Let $\Gamma_T(\theta)=-T^{-1}\partial_\theta^2\mathbb H_T(\theta)$ 
for $\theta\in\mathcal C[\Theta]$. 
Then 
\begin{equation}\label{171022-1} 
\Gamma_T(\theta) [u^{\otimes2}]
= 
\frac{1}{T}\int_0^T \bigg({\sf V}_0(\mathbb X(t),\theta)\otimes\mathbb X(t)^{\otimes2}\bigg)
[u^{\otimes2}]\sum_{i\in\mathbb I}dN^i_t.
\end{equation}

Take parameters $\alpha$, $\beta_1$, $\rho_1$ and $\rho_2$
so that 
\begin{equation} 
0<\beta_1<\frac{1}{2},\quad 0<\rho_1<\min\bigg\{1,\beta,\frac{2\beta_1}{1-\alpha}\bigg\},
\quad
0<2\alpha<\rho_2,
\quad 1-\rho_2>0,
\label{201903091202}
\end{equation}
where $\beta=\alpha/(1-\alpha)$. 
We have 
\begin{lemma}\label{170926-1} 
Suppose that $[A1]$ and $[A2]$ are satisfied. 
Let $p$ be any positive number. Then 
\begin{description}
\item[(i)] $  \sup_{T>1}\big\|\Delta_T\big\|_p<\infty$.
\item[(ii)] $  \sup_{T>1}\bigg\|\sup_{\theta\in\mathcal C[\Theta]}T^{\frac{1}{2}}%T^{\frac{1}{2}-\beta_2}
\big(\overline{\mathbb Y}_T(\theta)-\overline{\mathbb Y}(\theta)\big)\bigg\|_p<\infty$. 
\item[(iii)] 
$  \sup_{T>1}\bigg\|T^{-1}\sup_{\theta\in\mathcal C[\Theta]}|\partial_\theta^3\mathbb H_T(\theta)\big|
\bigg\|_p<\infty$. 
\item[(iv)] 
$  \sup_{T>1}E\big\|T^{\beta_1}|\Gamma_T-\Gamma|\big\|_p<\infty$. 
\end{description}
\end{lemma}
\begin{proof} 
Let $\rho(x,\theta)=\big(\rho^i(x,\theta)\big)_{i\in\mathbb I_0}$. 
Let $e_i=(\delta_{\alpha,i})_{\alpha\in\mathbb I_0}$ for $i\in\mathbb I$. 
Define $g_i\in\mathbb R^{\overline{i}}\otimes\mathbb R^{\overline{j}}$ $(i\in\mathbb I)$ by 
\begin{equation} 
g_i(x,\theta) = \big(1_{\{i\not=0\}}e_i-\rho(x,\theta)\big)\otimes x
\qquad (i\in\mathbb I)
\end{equation}
Then 
\begin{equation} 
\Delta_T = T^{-1/2}\sum_{i\in\mathbb I}\int_0^Tg_i(\mathbb X(t),\theta^*)d\tilde{N}^i_t,
\label{201903091157}
\end{equation}
where 
\begin{equation} 
\tilde{N}^i_t = N^i_t-\int_0^t \lambda^i(s,\vartheta^*)ds\qquad(i\in\mathbb I)
\end{equation}

We have 
\begin{equation} 
\lambda^i(t,\theta) \leq \lambda_0(t)\sum_{i\in\mathbb I}\exp\big(\mathbb X(t)[\vartheta^{*i}]\big),
\\
|g_i(\mathbb X(t),\theta)| \leq 2|\mathbb X(t)|
\end{equation}
and 
\begin{equation} 
|\log r^i(t,\theta)| = |\log \rho^i(\mathbb X(t),\theta)| \leq C_{\overline{i}}\>(1+|\mathbb X(t)||\theta|)
\end{equation}
for some constant $C_{\overline{i}}$ depending on $\overline{i}$. 

By the Burkholder-Davis-Gundy inequality, we obtain 
\begin{align} 
E\bigg[\bigg|T^{-1/2}\int_0^T h(\mathbb X(t))d\tilde{N}^i_t\bigg|^{2k}\bigg]
& \leq
C_kE\bigg[\bigg|T^{-1}\int_0^T h(\mathbb X(t))^2dN^i_t\bigg|^{k}\bigg]
\nonumber \\ & \leq
C_kE\bigg[\bigg|T^{-1}\int_0^T h(\mathbb X(t))^2d\tilde{N}^i_t\bigg|^{k}\bigg]
\nonumber \\ & \quad
+ C_kE\bigg[T^{-1}\int_0^T \big|h(\mathbb X(t))\big|^{2k}
\lambda^i(t,\vartheta^*)^k dt \bigg]
\nonumber \\ & \leq
C_kE\bigg[\bigg|T^{-1}\int_0^T h(\mathbb X(t))^2d\tilde{N}^i_t\bigg|^{k}\bigg]
\nonumber \\ & \quad
+ C_kE\bigg[\big|h(\mathbb X(0))\big|^{2k}
\lambda^i(0,\vartheta^*)^k  \bigg]
\end{align}
for any $k\in\mathbb N$ and any measurable function $h$ of at most polynomial growth. 
By Equation \eqref{201903091157} and induction, we obtain (i). 

Set 
\begin{equation} 
M_T(\theta) = 
\sum_{i\in\mathbb I}T^{-1}\int_0^T \log\frac{\rho^i(\mathbb X(t),\theta)}{\rho^i(\mathbb X(t),\theta^*)}d\tilde{N}^i_t
\end{equation}
and 
\begin{equation}
K_T(\theta) = T^{-1}\int_0^T f(\lambda_0(t),\mathbb X(t),\theta)dt,
\end{equation}
where 
\begin{align} 
f(w,x,\theta) & = 
\sum_{i\in\mathbb I}\log \frac{\rho^i(x,\theta)}{\rho^i(x,\theta^*)}\rho^i(x,\theta^*)
\Lambda(w,x)
\nonumber \\ & \quad
-E\bigg[\sum_{i\in\mathbb I}\log \frac{\rho^i(\mathbb X(0),\theta)}{\rho^i(\mathbb X(0),\theta^*)}
\rho^i(\mathbb X(0),\theta^*)\Lambda(\lambda_0(0),\mathbb X(0))\bigg]. 
\end{align}
Then 
\begin{equation} 
\overline{\mathbb Y}_T(\theta)-\overline{\mathbb Y}(\theta)
=
M_T(\theta)+K_T(\theta). 
\end{equation}
Similarly to the proof of (i), we obtain 
\begin{equation} 
\sum_{k=0,1}\sup_{\theta\in\mathcal C[\Theta]}\sup_{T>1}\|T^{1/2}\partial_\theta^kM_T(\theta)\|_p <\infty
\end{equation}
for every $p\geq2$. 
Moreover, applying Theorem 6.3 of \cite{Rio2017} under $[A1]$ and $[A2]$, 
we obtain 
\begin{equation} 
\sum_{k=0,1}\sup_{\theta\in\mathcal C[\Theta]}\sup_{T>1}\|T^{1/2}\partial_\theta^kK_T(\theta)\|_p <\infty
\end{equation}
for every $p\geq2$. Now Sobolev's embedding inequality in $W^{1,p}(\mathcal C[\Theta]) \hookrightarrow C(\mathcal C[\Theta])$ for $p>\mathsf p\vee2$ 
gives (ii). 
In a similar fashion, it is possible to prove (iii). 
The proof of (iv) is also similar to that of (ii), and rather simpler. 
\end{proof}

\begin{lemma}\label{170926-2} 
Suppose that $[A1]$-$[A3]$ are satisfied. Then 
\begin{description}\item[(i)] For any $L>0$, 
\begin{equation} 
\sup_{T>1}\sup_{r>0}\>r^L
P\bigg[\sup_{u\in\mathbb V_T(r)}\mathbb Z_T(u)
\geq \exp\big(-2^{-1}r^{2-(\rho_1\vee\rho_2)}\big)\bigg]
<
\infty,
\end{equation}
where $\mathbb V_T(r)=\{u\in\mathbb U_T;\>|u|\geq r\}$. 
\item[(ii)] $\mathbb Z_T$ admits a locally asymptotically normal representation 
\begin{equation}\label{290926-3} 
\mathbb Z_T(u) = 
\exp\bigg(\Delta_T[u]-\frac{1}{2}\Gamma[u^{\otimes2}]+r_T(u)\bigg)
\end{equation}
with $r_T(u)\to^p0$ as $T\to\infty$ for every $u\in\mathbb R^\mathsf p$ 
and $\Delta_T\to^d N_\mathsf p(0,\Gamma)$ as $T\to\infty$. 
\end{description}
\end{lemma}
\begin{proof} 
We will verify the conditions of 
Theorem 3 (c) of \cite{Yoshida2011} for the naturally extended random field 
$\mathbb H_T$ over $\mathcal C[\Theta]$. 
Condition $[A1'']$ therein holds according to Lemma \ref{170926-1} (iii) and (iv). 
Condition $[A4']$ therein is satisfied with $\beta_2=0$ under Equation \eqref{201903091202}. 
Lemma \ref{170926-1} (i) and (ii) ensures Condition $[A6]$ therein. 
Condition $[B1]$ therein is obvious and 
Condition $[B2]$ therein is (\ref{290925-7}). 
Therefore, by Theorem 3 of \cite{Yoshida2011}, we obtain 
\begin{equation} 
\sup_{T>1}\sup_{r>0}\>r^L
P\bigg[\sup_{u\in\mathcal C[\mathbb U_T]:|u|\geq r}\overline{\mathbb Z}_T(u)
\geq \exp\big(-2^{-1}r^{2-(\rho_1\wedge\rho_2)}\big)\bigg]
<
\infty, 
\end{equation}
which gives (i). 

The term $r_T(u)$ is defined by (\ref{290926-3}) for $u\in\mathbb U_T$. 
For each $u\in\mathbb R^\mathsf p$, for sufficiently large $T$, 
$r_T(u)$ admits the representation 
\begin{equation} 
r_T(u) = 
\int_0^1 (1-s)\bigg\{\Gamma[u^{\otimes2}]-\Gamma_T(\theta^\dagger_T(su))[u^{\otimes2}]
\bigg\}ds.
\end{equation}
Then Lemma \ref{170926-1} (iii) and (iv) verify the convergence $r_T(u)\to^p0$ 
as $T\to\infty$. 

For $u=(u^i_j)_{i\in\mathbb I_0,j\in\mathbb J}$ and $x=(x_j)_{j\in\mathbb J}$, 
\begin{align}
& \; \sum_{i\in\mathbb I}\bigg\{\sum_{i_1\in\mathbb I_0,j_1\in\mathbb J}
\big(\delta_{i,i_1}-\rho^{i_1}(x,\theta)\big)x_{j_1}u^{i_1}_{j_1}\bigg\}^2\rho^i(x,\theta)
\nonumber \\ = & \;
\sum_{\substack{i_1,i_2\in\mathbb I_0\\ j_1,j_2\in\mathbb J}}
\bigg\{\sum_{i\in\mathbb I}\big(\delta_{i,i_1}-\rho^{i_1}(x,\theta)\big)
\big(\delta_{i,i_2}-\rho^{i_2}(x,\theta)\big)\rho^i(x,\theta)\bigg\}
x_{j_1}x_{j_2}u^{i_1}_{j_1}u^{i_2}_{j_2}
\nonumber \\ = & \;
\sum_{\substack{i_1,i_2\in\mathbb I_0\\ j_1,j_2\in\mathbb J}}
\bigg\{\delta_{i_1,i_2}\rho^{i_1}(x,\theta)-\rho^{i_1}(x,\theta)\rho^{i_2}(x,\theta)\bigg\}
x_{j_1}x_{j_2}u^{i_1}_{j_1}u^{i_2}_{j_2}
\nonumber \\ = & \;
\big({\sf V}_0(x,\theta)\otimes x^{\otimes2}\big)[u^{\otimes2}].
\end{align}
Since it is assumed that $N^i$ ($i\in\mathbb I$) do not have common jumps, 
\begin{align}
& \;\bigg\langle T^{-1/2}\sum_{i\in\mathbb I}\int_0^\cdot g_i(\mathbb X(t),\theta^*)[u]
d\tilde{N}^i_t\bigg\rangle_T
\nonumber \\ = & \;
T^{-1}\sum_{i\in\mathbb I}\int_0^T\big(g_i(\mathbb X(t),\theta^*)[u]\big)^2
\rho^i(\mathbb X(t),\theta^*)\Lambda(\lambda_0(t),\mathbb X(t))dt
\nonumber \\ = & \;
T^{-1}\int_0^T\bigg({\sf V}_0(\mathbb X(t),\theta)\otimes \mathbb X(t)^{\otimes2}\bigg)[u^{\otimes2}]
\Lambda(\lambda_0(t),\mathbb X(t))dt. 
\end{align}
Under $[A1]$ and $[A2]$, the term on the right-hand side converges in probability 
to $\Gamma[u^{\otimes2}]$ as $T\to\infty$. 
The conditional type Lindeberg condition is easily verified by dividing the range $[0,T]$ of the integral \eqref{201903091157} into 
$\lceil T\rceil$ subintervals, and as a result, 
we see $\Delta_T\to^dN_\mathsf p(0,\Gamma)$, which concludes the proof of (ii). 
\end{proof}

It is possible to extend $\mathbb Z_T$ to $\mathbb R^\mathsf p$ so that the extension has a compact support and 
\begin{equation}
\sup_{u\in\mathbb R^\mathsf p\setminus\mathbb U_T}\mathbb Z_T(u)\leq\max_{u\in\partial\mathbb U_T}\mathbb Z_T(u).
\end{equation} 
We will denote this extended random field by the same $\mathbb Z_T$. 
Then $\mathbb Z_T$ is a random variable taking values in 
the Banach space 
$\hat{C}=\{f\in C(\mathbb R^\mathsf p);\>\lim_{|u|\to\infty}f(u)=0\}$ equipped with sup-norm. 
On some probability space, we prepare a random field 
\begin{equation} 
\mathbb Z(u) = \exp\bigg(\Delta[u]-\frac{1}{2}\Gamma[u^{\otimes2}]\bigg)
\qquad (u\in\mathbb R^\mathsf p)
\end{equation}
where $\Delta\sim N_\mathsf p(0,\Gamma)$. 
By Lemma \ref{170926-2} (ii), we have a finite-dimensional convergence 
\begin{equation}\label{290926-6}
\mathbb Z_T  \to^{d_f} \mathbb Z
\end{equation}
as $T\to\infty$. 

For $\delta>0$ and $c>0$, define $w_T(\delta,c)$ by 
\begin{equation} 
w_T(\delta,c) = \sup\bigg\{
\big|\log\mathbb Z_T(u_2)-\log\mathbb Z_T(u_1)\big|;\>u_1,u_2\in B_{c}, \> |u_2-u_1|\leq\delta\bigg\}
\end{equation}
for large $T$, where $B_c=\{u\in\mathbb R^\mathsf p;\>|u|\leq c\}$. 
Then by Lemma \ref{170926-1} (iii) and the definition of $\mathbb Z_T$, 
we have 
\begin{equation}\label{290926-7}
\lim_{\delta\downarrow 0}\limsup_{T\to\infty}P\big[w_T(\delta,c)>\epsilon\big]=0
\end{equation}
for every $\epsilon>0$ and $c>0$. 

According to e.g. Theorem 4 of \cite{Yoshida2011}, 
we obtain \eqref{170904-1} for $\hat{u}^M_T$ since 
Conditions $[C1]$ and $[C3]$ therein are ensured by (\ref{290926-7}) and 
by (\ref{290926-6}), respectively, and Condition $[C2]$ is trivial now. 

The properties (\ref{290926-7}) and (\ref{290926-6}) give 
functional convergence 
\begin{equation} 
\mathbb Z_T|_{B_c} \to^d \mathbb Z_T|_{B_c}\qquad \text{in}\quad C(B_c)
\end{equation}
as $T\to\infty$ for each $c>0$. 
Moreover, Lemma \ref{170926-2} already provided the PLD inequality for $\mathbb Z_T$. 
Thus  e.g. Theorem 10 of \cite{Yoshida2011} 
proves (\ref{170904-1}) for $\hat{u}^B_T$ 
once the estimate 
\begin{equation}\label{170926-8}
\sup_{T>1}E\bigg[\bigg(\int_{\mathbb U_T}\mathbb Z_T(u)du\bigg)^{-1}\bigg]
<\infty
\end{equation}
is established. However, (\ref{170926-8}) can be verified e.g. with Lemma 2 of \cite{Yoshida2011}. 
This completes the proof of Theorem \ref{170925-2}. 
\qed

\subsection{Proof of Proposition \ref{20180504-1}} 
Suppose that $\theta^*\not\in S_\mathbb{K}$. Then 
\begin{align}
\frac{1}{T}\big(C_T(S_\mathbb{K})-C_T(S_{\mathbb{K}^*})\big) & = -\frac{2}{T}\big(\mathbb{H}_T(\hat{\theta}_\mathbb{K})-\mathbb{H}_T(\theta^*)\big) +\frac{2}{T}\big(\mathbb H_T(\hat{\theta}_{\mathbb K^*})-\mathbb H_T(\theta^*)\big) \nonumber
\\ & \quad + \frac{a_T}{T}\big(\mathsf{d}(S_\mathbb{K})-\mathsf{d}(S_\mathbb{K}^*)\big) \nonumber
\\ & \geq 2\chi_0\inf_{\theta\in S_\mathbb{K}}\big|\theta-\theta^*\big|^2 +o_p(1)
\end{align}
by Lemma \ref{170926-1} (ii) and (\ref{290925-7}) applied to the sub-models $S_\mathbb K$ and $S_{\mathbb K^*}$ in place of $\Theta$. Then (i) holds since $\inf_{\theta\in S_\mathbb{K}}\big|\theta-\theta^*\big|>0$. 

Next, suppose that $\theta^*\in S_\mathbb{K}$ and $S_\mathbb{K}\not=S_{\mathbb{K}^*}$. Obviously, $\mathsf{d}(S_\mathbb{K})>\mathsf{d}(S_{\mathbb{K}^*})$. Denote by $\Gamma_\mathbb{K}$ the matrix consisting of the elements of $\Gamma$ with indices in $\mathbb{K}$. Then $[A3]$ implies $\det\Gamma_\mathbb{K}>0$. Thus, we can obtain the same results as in Theorem \ref{170925-2} for $\hat{\theta}_\mathbb{K}$. In particular, 
\begin{equation}
2\big(\mathbb{H}_T(\hat{\theta}_\mathbb{K})-\mathbb{H}_T(\theta^*)\big) = \Gamma_\mathbb{K}\big[(\hat{u}_\mathbb{K})^{\otimes2}\big]+o_p(1) =  O_p(1)
\end{equation}
for both QMLE and QBE for the model $S_\mathbb K$ since $\theta^*\in S_\mathbb K$, where $\hat{u}_\mathbb{K}=\sqrt{T}(\hat{\theta}_\mathbb{K}-\theta\*)$. This is also valid for the model $S_{\mathbb K^*}$ with  
$\hat{\theta}_{\mathbb K^*}$.
Therefore, 
\begin{equation}
C_T(S_\mathbb{K})-C_T(S_{\mathbb{K}^*}) = O_p(1)+\big(\mathsf{d}(S_\mathbb{K})-\mathsf{d}(S_{\mathbb{K}^*})\big)a_T \>\to^p\>\infty
\end{equation}
as $T\to\infty$, which completes the proof. 
\qed

\section{List of stocks}
\label{sec:ricList}
Table \ref{table:ricList} lists all the stocks investigated in the paper.

\begin{table}
\begin{center}
\footnotesize
\begin{tabular}{cccc}
\hline
\multirow{2}{*}{RIC} & \multirow{2}{*}{Company} & \multirow{2}{*}{Sector} & Number of trading
\\ & & & days in sample \\ \hline
AIRP.PA & Air Liquide & Healthcare / Energy & 239 \\
BNPP.PA & BNP Paribas & Banking & 224 \\ 
EDF.PA & Electricite de France & Energy & 237 \\ 
LAGA.PA & Lagardère & Media & 143 \\ 
CARR.PA & Carrefour & Retail & 230 \\ 
BOUY.PA & Bouygues & Construction / Telecom & 229 \\ 
ALSO.PA & Alstom & Transport & 230 \\
ACCP.PA & Accor & Hotels & 228 \\ 
ALUA.PA & Alcatel & Networks / Telecom & 235 \\ 
AXAF.PA & Axa & Insurance & 237 \\ 
CAGR.PA & Crédit Agricole & Banking & 236 \\ 
CAPP.PA & Cap Gemini & Technology Consulting & 233 \\ 
DANO.PA & Danone & Food & 230 \\ 
ESSI.PA & Essilor & Optics & 229 \\ 
LOIM.PA & Klepierre & Finance & 222 \\ 
LVMH.PA & Louis Vuitton Moët Hennessy & Luxury & 234 \\ 
MICP.PA & Michelin & Tires & 230 \\ 
OREP.PA & L'Oréal & Cosmetics & 234 \\ 
PERP.PA & Pernod Ricard & Spirits & 225 \\ 
PEUP.PA & Peugeot & Automotive & 152 \\ 
PRTP.PA & Kering & Luxury & 228 \\ 
PUBP.PA & Publicis & Communication & 224 \\ 
RENA.PA & Renault & Automotive & 229 \\ 
SAF.PA & Safran & Aerospace / Defense & 233 \\ 
TECF.PA & Technip & Energy & 226 \\ 
TOTF.PA & Total & Energy & 233 \\ 
VIE.PA & Veolia & Energy / Environment & 235 \\ 
VIV.PA & Vivendi & Media & 235 \\ 
VLLP.PA & Vallourec & Materials & 229 \\ 
VLOF.PA & Valeo & Automotive & 222 \\ 
SASY.PA & Sanofi & Healthcare & 230 \\ 
SCHN.PA & Schneider Electric & Energy & 225 \\ 
SGEF.PA & Vinci & Construction & 230 \\ 
SGOB.PA & Saint Gobain & Materials & 235 \\ 
SOGN.PA & Société Générale & Banking & 230 \\ 
STM.PA & ST Microelectronics & Semiconductor & 228 \\ \hline
\normalsize
\end{tabular}
\caption{List of stocks investigated in this paper. Sample consists of the whole year 2015, representing roughly 230 trading days for all stocks except LAGA.PA and PEUP.PA which are missing roughly 70 trading days.}
\label{table:ricList}
\end{center}
\end{table}

\newpage

\bibliographystyle{authordate1}
\bibliography{Cox}

\end{document}